\documentclass[a4paper]{article}

\usepackage{amsfonts}
\usepackage{amssymb,amsthm,amsmath}
\usepackage{enumitem}
\usepackage{multicol}
\usepackage{color}

\newcommand{\id}{\mathop{\mathrm{id}}}
\newcommand{\stab}{\mathop{\mathrm{stab}}\nolimits}
\newcommand{\supp}{\mathop{\mathrm{supp}}}

\newcommand{\Sym}{\mathop{\mathrm{Sym}}}

\newtheorem{prop}{Proposition}
\newtheorem{deft}[prop]{Definition}

\usepackage{orcidlink}
\usepackage{hyperref}

\begin{document}
\title{Irreducible Rules and Equivalence Classes of One-dimensional Cellular Automata}

\author{
    Martin Schaller\,\orcidlink{0009-0000-1939-8475}\thanks{martin.roman.schaller@gmail.com} \\
    \small{Vienna, Austria}
    \and
    Karl Svozil\,\orcidlink{0000-0001-6554-2802}\thanks{karl.svozil@tuwien.ac.at, \url{http://tph.tuwien.ac.at/~svozil}} \\
    \small{Institute for Theoretical Physics, TU Wien,} \\
    \small{Wiedner Hauptstrasse 8-10/136, 1040 Vienna, Austria}
}

\date{\today}
\maketitle

\begin{abstract}
One-dimensional cellular automata are discrete dynamical systems that operate on an infinite lattice of sites and
are characterized by the locality
and uniformity of their update rule.
Permutations of the state set and isometric transformations of the lattice induce symmetry transformations
on the set of local rules and the set of global maps of cellular automata, resulting
in a partitioning of the set of cellular automata into equivalence classes.
The concept of an irreducible local rule that depends on all
its coordinates is used to analyse the equivalence classes and
results on the number of equivalence classes of irreducible binary local rules and binary global maps are presented.
Finally, another symmetry operator based on the scaling of neighbourhoods is introduced and
the change in the number of equivalence classes is analysed.
\end{abstract}

\section{Introduction}

One-dimensional cellular automata (CAs) are discrete dynamical systems in which
space, time, and state are all taken to be discrete. They consist of a bi-infinite
array of identical cells indexed by the integers, each of which takes its value from
a finite state set (or alphabet). Time evolves in discrete steps according to a fixed,
local update rule that is applied synchronously and in parallel to all cells. This
combination of locality, uniformity, and discreteness makes CAs a natural model
for spatially extended computation and for complex dynamical behaviour arising
from simple microscopic laws.

Formally, a one-dimensional CA is specified by three ingredients: a finite state
set $\mathcal{A}$, a finite neighbourhood $N \subset \mathbb{Z}$, and a local rule
$f : \mathcal{A}^{N} \to \mathcal{A}$. For each site $i \in \mathbb{Z}$, the neighbourhood $N$ is
first translated to $i+N = \{ i+j : j \in N\}$, and the new state of cell $i$ at the
next time step is given by applying $f$ to the tuple of states at positions in
$i+N$. Repeating this procedure in parallel for all cells yields a new
configuration, that is, a new bi-infinite sequence of cell states. In this way, the
local rule $f$ and the neighbourhood $N$ determine a \emph{global map}
$\Phi_{f}^{N} : \mathcal{A}^{\mathbb{Z}} \to \mathcal{A}^{\mathbb{Z}}$, which associates each configuration with
its successor. Translation-invariance of the dynamics is encoded in the fact that
$\Phi_{f}^{N}$ commutes with the (left) shift $\sigma$ on configurations.

Conceptually, it is often advantageous to separate the local and global viewpoints.
On the one hand, the local rule $f$ can be viewed as a finite lookup table, and
its combinatorial structure determines many of the observable behaviours of the
CA. On the other hand, from the topological dynamical systems perspective,
the CA is primarily the global map $\Phi : \mathcal{A}^{\mathbb{Z}} \to \mathcal{A}^{\mathbb{Z}}$. The
Curtis--Hedlund--Lyndon theorem (see, e.g., \cite{DBLP:journals/mst/Hedlund69,Lind_Marcus_1995})
bridges these viewpoints by characterizing CA global maps as exactly those maps
on $A^{\mathbb{Z}}$ that are continuous (with respect to the product topology) and
commute with the shift. Thus any such $\Phi$ admits at least one representation
as $\Phi = \Phi_{f}^{N}$ for a suitable neighbourhood $N$ and local rule $f$.

A key theme in the theory of CAs is that many different syntactic descriptions,
comprising neighbourhoods and local rules, can represent essentially the same
dynamics. This redundancy is especially pronounced when symmetry transformations
of the underlying space and state set are taken into account. In the one-dimensional
case, there are two main sources of symmetry:

\begin{enumerate}
  \item \emph{Permutations of the state set.} Any bijection
        $\nu : \mathcal{A} \to \mathcal{A}$ induces a relabeling of configurations by acting
        on each cell independently. If we simultaneously relabel the entries in
        the lookup table of $f$ in a consistent way, the global dynamics is
        unchanged up to this relabeling. In other words, CAs that differ only
        by such permutations should be regarded as equivalent.

  \item \emph{Geometric transformations of the lattice.} The group of
        isometries of the Euclidean line consists of translations and reflections, see e.g.,~\cite{martin2012transformation}.
        On the discrete lattice $\mathbb{Z}$, these act by integer shifts and by
        reversing the direction of the lattice. A reflection transforms a
        neighbourhood $N$ into $-N$, and a corresponding transformation of
        the local rule yields a CA whose global dynamics is the mirror image of
        the original one. Translations of the lattice correspond to shifting the
        neighbourhood and to conjugating the global map by the shift.
\end{enumerate}

Beyond pure isometries, we will also consider \emph{uniform scaling} of the
lattice: given a neighbourhood $N$, we may form a ``stretched'' neighbourhood
$pN = \{pj : j \in N\}$ for some integer scaling factor $p$. Unlike translations
and reflections, this mapping is not onto, due to the discrete nature of the
lattice. Nevertheless, at the level of global maps one can relate a CA with
neighbourhood $pN$ to another CA with neighbourhood $N$ by a suitable
decimation and interleaving of configurations. This introduces a further
symmetry operator, the \emph{scaling operator}, on the space of CA global maps.

Taken together, permutations of the state set, reflections and translations of
the lattice, and scaling of neighbourhoods generate a rich family of symmetry
operators acting either on local rules or on global maps. These symmetries define
an equivalence relation on the set of all CAs over a fixed alphabet $A$, and
hence partition this set into \emph{equivalence classes}. Two CAs are deemed
equivalent if one can be mapped onto the other by a finite composition of these
symmetry operations. Determining these equivalence classes and their
cardinalities is the principal aim of this work.

Historically, most attention has been restricted to a smaller symmetry group:
permutations of the state set and reflections of the lattice, acting only on local
rules of a fixed arity. In this restricted setting the group is finite, and
classical tools from group action theory (in particular, Burnside's lemma) can
be applied to count equivalence classes. In contrast, once we move to global
maps and include translations, the relevant symmetry group becomes
infinite, while scaling is only a partial mapping of the global rule space on itself,
and new ideas are required.

A central technical device in this paper is the notion of an \emph{irreducible}
local rule. Given $f : \mathcal{A}^{n} \to \mathcal{A}$, we say that $f$ is independent of its
$j$-th argument if changing only the $j$-th entry of the input never changes
the output. Otherwise, $f$ is said to \emph{depend} on index $j$. The
\emph{support} of $f$ is the set of indices on which $f$ depends. We call $f$
\emph{irreducible} if its support has full size $n$, i.e., if it depends on
all arguments.

Many local rules are reducible in this sense: they can be written as the
composition of an irreducible rule of smaller arity with a projection that discards
irrelevant coordinates. At the level of global maps, this reducibility translates
into the familiar fact that the same global map can often be implemented by
different pairs $(N,f)$: one can enlarge the neighbourhood by adding
redundant sites on which $f$ does not actually depend. For instance, the binary
identity map can be realized by infinitely many CAs with different neighbourhoods,
all of which simply copy the state at the central site.

The key observation is that if we insist on representing a global map by a CA
whose local rule is irreducible, then the representation becomes unique. In other
words, every CA global map has a \emph{canonical} description in terms of an
irreducible local rule and a minimal neighbourhood, and this description is
unique up to equality of neighbourhood and local rule. Formally, there is a
bijection between the set of continuous, shift-commuting maps
$\Phi : \mathcal{A}^{\mathbb{Z}} \to \mathcal{A}^{\mathbb{Z}}$ and the set of irreducible CAs over
$\mathcal{A}$. This provides a convenient way to pass back and forth between the
dynamical (global) and combinatorial (local) viewpoints.

When studying equivalence under symmetry transformations, one is naturally
led to consider the \emph{stabilizer} of a local rule (or global map), that is, the
subgroup of symmetries that leave it invariant. For example, in the binary case,
the group generated by state complementation and reflection is isomorphic to the
Klein four-group. Each equivalence class of rules under this group action has an
associated stabilizer subgroup, and we can speak of the \emph{type} of the class,
meaning the conjugacy class of its stabilizer.

For local rules of fixed arity, this perspective is classical: the finite symmetry
group acts on the finite set of local rules, producing a partition into
equivalence classes which can be further refined by type. In the binary case,
these types correspond to different combinations of reflection and complementation
symmetries possessed by the rule. Part of the present work revisits this
setting, but with the crucial restriction to \emph{irreducible} rules, and derives
explicit counts for the number of equivalence classes of each type.

The main novelty of this paper lies in lifting the entire symmetry analysis to
the level of global maps. Here, translations of the lattice cannot be avoided:
even if two local rules are related only by reflection, the induced global maps
are related by a composition of reflection and shift.
The symmetry transformations of permutation, reflection, and shifting form
an infinite group $\mathcal{T}_{\mathcal{A}}$ of symmetry operators acting
on the space of all CA global maps over the alphabet $A$.
Scaling is a further important symmetry operator, although due the discreteness of
the integers, it is not a tranformation but only a partial mapping of the global rule
spac on itself.
The resulting
equivalence relation is strictly coarser than the one induced by permutations and
reflections on local rules, and we show how to exploit the classification of
irreducible rules to understand this richer quotient structure.

Within this framework, this paper makes several contributions:

\begin{itemize}
  \item We formalize a uniform family of symmetry operators on the space of all
        CA global maps over a given alphabet $\mathcal{A}$, arising from permutations
        of $\mathcal{A}$, reflections and translations of the integer lattice, and uniform
        scaling of neighbourhoods. These operators are naturally related to the
        isometries and similarities of the Euclidean line.

  \item We develop a group-theoretical description of the induced equivalence
        relation on global maps, via an infinite group $T_{A}$ of symmetry
        transformations. This allows us to speak systematically about
        equivalence classes of CAs, not just at the level of local rules but at
        the level of the dynamics $\Phi : \mathcal{A}^{\mathbb{Z}} \to \mathcal{A}^{\mathbb{Z}}$.

  \item We exploit the notion of an irreducible local rule. We
        show that every CA global map admits a unique representation by an
        irreducible local rule and neighbourhood, thereby establishing a one-to-one
        correspondence between irreducible CAs and global maps.

  \item For a fixed alphabet and arity, we derive a closed formula for the
        number of irreducible local rules, using an inclusion--exclusion argument
        over subsets of coordinates.

  \item Specializing to the binary case, we determine the equivalence classes of
        irreducible local rules of a given arity under the action of the symmetry
        group generated by state complementation and reflection, and we classify
        these classes by type (stabilizer subgroup).

  \item Using this classification, we compute the number of equivalence classes
        of binary global maps induced by a fixed contiguous neighbourhood, when
        all symmetries given by state permutations, reflection, and shifts are
        taken into account. This yields, in particular, exact counts for small
        neighbourhoods that previously could only be obtained by brute-force
        enumeration.

  \item We introduce a scaling operator on global maps, based on decimating
        and interleaving configurations, and analyse how it further merges
        equivalence classes. Applying this to binary CAs with neighbourhood
        $[-1,1]$ (the so-called \emph{elementary} CAs), we show that the
        256 rules split into 81 equivalence classes under the full symmetry
        group and scaling.
\end{itemize}

Many of these results are new, and others provide structural explanations for
numerical counts that were previously known only from extensive computations.
Throughout, we emphasize how irreducible rules provide a natural canonical
framework in which to phrase and solve these classification problems.

Equivalence classes of one-dimensional CAs under reflection and permutation
have been studied in several foundational works. Wolfram \cite{Wolfram1983}
introduced the class of binary CAs with the symmetric contiguous neighbourhood
of size three, now commonly called \emph{elementary} cellular automata. In
\cite{wolfram-1986}, he gave a table partitioning the 256 elementary rules into
88 equivalence classes under reflection and state permutation. A later
mathematical derivation of the number of equivalence classes of elementary CAs
was provided by Li and Packard \cite{Li1990TheSO}.

Cattaneo et al.\ \cite{CATTANEO19971593} systematically
studied transformations on the one-dimensional CA rule space, including
symmetries of binary rules. Among their results is a general formula (based on
Burnside's lemma) for the number of equivalence classes of binary CAs with an
odd neighbourhood size, under the finite group generated by reflection and
permutations of the state set. Schaller and Svozil \cite{SCHALLER2025100298}
extended this line of work by determining the number of equivalence classes of
two- and three-state CAs for all contiguous neighbourhoods and by classifying
these classes by type (stabilizer subgroup).

Shifting and shift-equivalence of local rules has been investigated by
Acerbi et al.~\cite{10.1007/978-3-540-73001-9_1},
Ruivo et al.\
\cite{RUIVO2018280}, and by Garc{\'i}a-Morales \cite{GARCIAMORALES2013276}, who
studied how shifting the neighbourhood indices can generate further equivalence
between CA rules.
Ruivo et al. introduced also the notion of irreducibility of local rules.
More general symmetry transformations acting on CA rules
have been considered in, for example, Castillo-Ram{\'i}rez and Gadouleau~\cite{CASTILLORAMIREZ2020104533}, and
these ideas have been generalized to CAs defined over arbitrary groups rather
than just $\mathbb{Z}$; see Castillo-Ram{\'i}rez et al.\ \cite{Castillo-Ramirez03072023}.

Symmetry-induced equivalence is only one, very basic, form of CA classification.
More broadly, the classification of CAs by dynamical or computational
properties is a central topic in CA theory, and many frameworks have been
proposed; see the survey by Vispoel et al.\ \cite{VISPOEL2022133074}. A
noteworthy alternative equivalence notion is based on topological conjugacy:
Epperlein \cite{epperlein2015classification, epperlein2017} classified elementary CAs up to
topological conjugacy, obtaining 83 equivalence classes and demonstrating that
this relation identifies deeper structural similarities than simple permutation
and reflection.

The remainder of this paper is organized as follows.
In Section~\ref{sec:CA} we recall the formal definition of a one-dimensional CA in terms of
a local rule and a neighbourhood, and we relate this to the global-map
characterization given by the Curtis--Hedlund--Lyndon theorem. We introduce
notation for the spaces of local rules and global maps over a fixed alphabet and
neighbourhood, and briefly recall Wolfram's coding scheme for enumerating
binary local rules.

Section~\ref{sec:irr} is devoted to irreducible local rules. We define dependence on
coordinates, introduce the support of a rule, and formalize the notion of
irreducibility. We show that every local rule can be written uniquely as the
composition of an irreducible rule with a projection, and we prove that every
global map has a unique representation by an irreducible CA. This yields a
one-to-one correspondence between irreducible CAs and global maps.

In Section~\ref{sec:loc} we focus on symmetry transformations acting on local rules. We
define the actions of state permutations and reflection on the local rule space,
and we recall the resulting group action and equivalence relation on rules of
fixed arity. Restricting attention to binary alphabets, we then derive explicit
formulae for the number of equivalence classes of \emph{irreducible} local rules
of given arity and for the number of classes of each stabilizer type. We also
illustrate the resulting classification for small neighbourhood sizes.

Section~\ref{sec:glob} lifts the symmetry analysis to the level of global maps. We first show
how permutation, reflection, and shift act on global maps and prove that these
actions are well-defined, even when different local rules represent the same
global map. We then analyse the equivalence classes of binary global maps
induced by the same contiguous neighbourhood under the full symmetry group
generated by permutations, reflection, and shifts. In the final part of the section
we introduce the scaling operator, show how it relates neighbourhoods that are
integer multiples of one another, and determine the resulting reduction in the
number of equivalence classes, with particular attention to the elementary
binary CAs.

Section~\ref{sec:extended} places our geometric symmetries in the broader context of topological dynamics, discussing their relationship to topological conjugacy and time-reversal symmetry.

Section~\ref{sec:sum} collects our conclusions and summarizes the main results, emphasizing
the role of irreducible rules as canonical representatives and the impact of
including scaling in the symmetry analysis. We also briefly discuss how our
framework relates to, and complements, other classification approaches such as
topological conjugacy.

\section{Cellular Automata}
\label{sec:CA}

If $A$ is any set, we denote the size (or cardinality) of $A$ by $|A|$.
We write $\mathcal{P}(A)$ for the power set of $A$, the set of all subsets of $A$.
The states of a CA are represented by symbols from a finite set $\mathcal{A}$, also called an alphabet.
A word $w=x_1x_2\ldots x_{m}$ over the alphabet $\mathcal{A}$ is a finite sequence of symbols
from $\mathcal{A}$ juxtaposed, and we write $|w|$ to denote the  the length of the sequence,
that is $|x_1x_2\ldots x_{m}| = m$
(the notation $|.|$ denotes both the size of a set and the length of a word).
We write $\mathcal{A}^m$ for the set of all words of length $m$ over the alphabet $\mathcal{A}$.
We denote the integer interval $\{p, p+1, \ldots, q\}$ by $[p,q]$.
If $q < p$, then $[p,q]$ is the empty set.
A configuration $x$ is a bi-infinite sequence over
the alphabet $\mathcal{A}$, formally defined as a map $\mathbb{Z} \to \mathcal{A}$.
The $i$-th element, $i \in \mathbb{Z}$, of a configuration $x$ is denoted by $x_i$.


\begin{deft}
\label{def:CA}
A one-dimensional CA $A$ is a triple $(\mathcal{A}, N, f)$, where \\
$\mathcal{A}$ is a finite set of states; \\
$N$ is the neighbourhood, a finite set of integers of size $n$; \\
$f$ is the local rule, a function from $\mathcal{A}^n$ to $\mathcal{A}$. \\
Let $N=\{j_1,j_2, \ldots,j_{n}\}$ such that $j_1 < j_2 < \cdots < j_{n}$.
The neighbourhood $N$ and the local rule $f$ induce the global map
$\Phi^A : \mathcal{A}^\mathbb{Z} \rightarrow \mathcal{A}^\mathbb{Z}$ on the set of configurations, defined by
$\Phi^N_f(x)_{i} = f(x_{i+j_1}x_{i+j_{2}}\ldots x_{i+j_{n}})$.
\end{deft}
We will also use the notation $\Phi^N_f$ for the global map, if the state set $\mathcal{A}$ is given.
If the CA is initialised with the configuration $x$, the CA computes in one step
the configuration $\Phi^N_f(x)$.
In case of $\mathcal{A} = \{0,1\}$, the CA, the local rule, and the global map, are all called binary.
The shift map (or operator) on $\mathcal{A}^\mathbb{Z}$ maps a configuration $x$ to the configuration $\sigma(x)$
whose $i$th coordinate is $\sigma(x)_i = x_{i+1}$.
It follows from the definition of the CA, that the global map commutes with the shift operator, that is
$\Phi^N_f(\sigma(x)) = \sigma(\Phi^N_f(x))$.

Let $x$ and $y$ be configurations of $\mathcal{A}^\mathbb{Z}$ and put
\[
d(x,y) = \begin{cases}
2^{-p} & \text{if $ x \ne y$ and $p$ is maximal so that $x_i = y_i$ for $-p \le i \le p$},\\
0 & \text{if $x = y$}.
\end{cases}
\]
The map $d$ is a metric on $\mathcal{A}^\mathbb{Z}$, which makes
$\mathcal{A}^\mathbb{Z}$ a topological space.
The following proposition, see \cite{DBLP:journals/mst/Hedlund69}, shows that CAs can also be defined by topological means.
\begin{prop}[Curtis-Hedlund-Lyndon Theorem]
A map of $\mathcal{A}^\mathbb{Z}$ into $\mathcal{A}^\mathbb{Z}$ is the global map of a CA if and only if
the map is continuous and commutes with the shift operator.
\end{prop}
A very accessible presentation of this theorem and symbolic dynamics in general can
be found in~\cite{Lind_Marcus_1995}.
We call the set $L_\mathcal{A}(n) = \{f \, | \, f : \mathcal{A}^n \to \mathcal{A} \}$ the local rule space of
arity $n$ over the alphabet $\mathcal{A}$.
If we set $k = |\mathcal{A}|$, the size of $L_{\mathcal{A}}(n)$ is $k^{k^n}$.
The set $G_\mathcal{A}(N) = \{ \Phi^N_f \, | \, f \in L_\mathcal{A}(|N|) \}$ is called the global rule
space on the neighbourhood $N$ over the alphabet $\mathcal{A}$ .
Wolfram~\cite{wolfram-1986} introduced a commonly used coding schema to enumerate the local rule space $L_{\{0,1\}}(n)$.
Let $f \in L_{\{0,1\}}(n)$ be a binary local rule.
The Wolfram code $c(f)$ of $f$ is the nonnegative integer
\[
c(f) = \sum_{a_1\ldots a_{n} \in \{0,1\}^n} f(a_1\ldots a_{n}) \times 2^{\sum_{i=1}^{n} a_i 2^{n-i}}.
\]
This summation uses the left-to-right most significant bit (MSB) convention.
If $f: \{ 0, 1 \}^n \to \{ 0, 1 \} $ is a local rule,
we define $W^n_{c(f)} = f$.
We denote the constant functions $f() = 0$ and $g() = 1$, both of arity 0, by
$W^0_0$ and $W^0_1$, respectively.
If this coding schema is used in combination with a global map, the arity of $f$ can be derived
from the neighbourhood, for instance, we can write
$\Phi^{[-1,1]}_{W_{110}}$ instead of $\Phi^{[-1,1]}_{W^3_{110}}$.

A global map of a CA can have different representations in terms of
local rules and neighbourhoods.
The binary identity map $\id: \{0,1\}^\mathbb{Z} \to \{0,1\}^\mathbb{Z}$,
$\id(x) = x$ serves as a trivial example.
For each neighbourhood $N$ that contains the integer 0,
the map $\id$ has a different representation, so that for instance
\[
\id = \Phi^{\{0\}}_{W_2} = \Phi^{[0,1]}_{W_{12}} = \Phi^{[-1,1]}_{W_{105}}.
\]
More general, if $N_2 \subseteq N_1$ then $G_\mathcal{A}(N_2) \subseteq G_\mathcal{A}(N_1)$.
In the next subsection we will see that each global map of a CA
has a unique canonical representation.

\section{Irreducibility}
\label{sec:irr}
A local rule $f : \mathcal{A}^n \to \mathcal{A}$ is said to be dependent on index $j$, $1 \le j \le n$,
if there is a word $a_1\ldots a_{j-1}a_{j+1}\ldots a_{n} \in \mathcal{A}^{n-1}$ and
symbols $b, c \in \mathcal{A}$ such that
$f(a_1\ldots a_{j-1}ba_{j+1}\ldots a_{n}) \ne f(a_1\ldots a_{j-1}ca_{j+1}\ldots a_{n})$,
otherwise $f$ is said to be independent of the index $j$.
The local rule $f: \mathcal{A}^n \to \mathcal{A}$ is said to be irreducible
if $f$ depends on all its indices, otherwise
$f$ is said to be reducible.
A CA $(\mathcal{A}, N, f)$ is said to be irreducible or reducible, if its local rule $f$ is
irreducible or reducible, respectively.
Define
 \[
\supp(f) = \{ j \, | \, \mbox{$f$ depends on coordinate $j$} \}.
\]
Then the local rule $f: \mathcal{A}^n \to \mathcal{A}$ is irreducible if and only if
$|\supp(f)| = n$.

A map $h: \mathcal{A}^n \to \mathcal{A}^m$, $m \le n$, is called a projection if
there are integers $p_1, \ldots, p_{m}$ with $1 \leq p_1 < \ldots < p_{m} \leq n$ such that
$h(a_1\ldots a_{n}) = a_{p_1}\ldots a_{p_{m}}$ for all words $a_1\ldots a_{n} \in \mathcal{A}^n$.
The set $\{p_1, \ldots, p_{m}\}$ is called the index set of the projection.
The map $H$ associates a projection with its index set, that is $H(h) = \{p_1, \ldots, p_{m}\}$.
The number of distinct projections $h : \mathcal{A}^n \to \mathcal{A}^m$ is given by
the binomial coefficient $\binom{n}{m}$.
If $f: \mathcal{A}^n \to \mathcal{A}$ is a local rule and $|\supp(f)| = m$,
then there is an irreducible local rule $g: \mathcal{A}^m \to \mathcal{A}$ and a projection
$h : \mathcal{A}^n \to \mathcal{A}^m$ such
that $f = g \circ h$.
This representation is unique.
This means that if $f = g^\prime \circ h^\prime$, $g^\prime$ is irreducible, and $h^\prime$ is a projection,
then $g = g^\prime$ and $h = h^\prime$.
Define $\overline{L}_\mathcal{A}(n) = \{f \in L_\mathcal{A}(n) \, | \, \mbox{$f$ is irreducible}  \}$ and
$\overline{G}_\mathcal{A}(N) = \{ \Phi_f^N \ | \ f \in \overline{L}_\mathcal{A}(|N|) \}$.
\begin{prop}
\label{prop:uniqueness}
If $A=(\mathcal{A}, N_1, f_1)$ and $B=(\mathcal{A}, N_2, f_2)$ are two irreducible CAs,
then $\Phi^A = \Phi^B$ if and only if $A=B$, that is if $N_1 = N_2$ and $f_1 = f_2$.
\end{prop}
\begin{proof}
If $N_1 = N_2$, then $\Phi^A = \Phi^B$ implies $f_1 = f_2$.
We therefore assume that $N_1 \neq N_2$.
Consequently, $N_1 \setminus N_2$ is nonempty, or $N_2 \setminus N_1$ is nonempty.
Without loss of generality, suppose that $N_1 \setminus N_2$ is nonempty.
Let $N_1 = \{j_1, \ldots, j_{n} \}$ and suppose $j_r \in N_1 \setminus N_2$.
The irreducibility of $f_1$ implies that there are words
$w_1 = a_1 \ldots a_{r-1}ba_{r+1} \ldots a_{n}$ and
$w_2 =a_1 \ldots a_{r-1}ca_{r+1} \ldots a_{n}$ in $\mathcal{A}^n$
such that $f_1(w_1) \neq f_1(w_2)$.
Construct two configurations $x$ and $y$ as follows.
Fix a symbol $d \in \mathcal{A}$.
If $i \not\in N_1$ set $x_i = y_i = d$.
If $j_i \in N_1$ and $j_i \neq j_r$ set $x_{j_i} = y_{j_i} = a_i$.
Set $x_{j_r} = b$ and $y_{j_r} = c$.
From the construction of $x$ and $y$ follows that $\Phi^A(x)_0 = f_1(w_1) \neq f_1(w_2) = \Phi^A(y)_0$ and
$\Phi^B(x)_0 = \Phi^B(y)_0$.
Thus $\Phi^A(x) \neq \Phi^B(x)$ or $\Phi^A(y) \neq \Phi^B(y)$.
\end{proof}
Proposition~\ref{prop:uniqueness} in conjunction with the Curtis-Hedlund-Lyndon Theorem imply the next proposition.
\begin{prop}
\label{prop:eq-cont}
If  $\Psi: \mathcal{A}^\mathbb{Z} \to \mathcal{A}^\mathbb{Z}$ is a continuous map that commutes with the shift operator $\sigma$,
then there is exactly one irreducible CA $A$ such that $\Psi = \Phi^A$.
\end{prop}
We denote the set of all continuous maps on $\mathcal{A}^\mathbb{Z}$ that commute with the shift operator by
$\mathcal{C}_\mathcal{A}$.
It is equal to the set of all global maps induced by CAs over the alphabet $\mathcal{A}$, and
there is a one-to-one correspondence between this set and all irreducible CAs over $\mathcal{A}$.
\begin{prop}
If $k=|\mathcal{A}|$, the size of $\overline{L}_\mathcal{A}(n)$ is given by $\sum_{i=0}^n (-1)^{n-i} \binom{n}{i} k^{k^i}$.
\end{prop}
\begin{proof}
Fix the alphabet $\mathcal{A}$ and put $k = |\mathcal{A}|$.
Set $a_n = |L_\mathcal{A}(n)| = k^{k^n}$ and $b_n = |\overline{L}_\mathcal{A}(n)|$.
Let $j \leq n$ and $f \in \overline{L}_\mathcal{A}(j)$.
Then $f$ contributes $\binom{n}{j}$ times to $a_n$.
We obtain the recurrence relation
\[
a_n = \sum_{j=0}^n \binom{n}{j} b_j.
\]
We plug the asserted equation
\[
b_j = \sum_{i=0}^j (-1)^{j-i} \binom{j}{i} a_i
\] into this relation and get
\begin{eqnarray*}
a_n = \sum_{j=0}^n \binom{n}{j} \sum_{i=0}^j (-1)^{j-i} \binom{j}{i} a_i
= \sum_{i=0}^n a_i \sum_{j=i}^n (-1)^{j-i} \binom{n}{j} \binom{j}{i}.
\end{eqnarray*}
The inner sum is known to satisfy the identity $\sum_{j=i}^n (-1)^{j-i} \binom{n}{j} \binom{j}{i} = \delta_{in}$, where $\delta_{in}$ is the Kronecker delta (equal to 1 if $i=n$, and 0 otherwise) \cite[p.~169]{Graham94}.
Thus the right hand side collapses to $a_n$, verifying the formula.
\end{proof}

\begin{table}
\begin{center}
\caption{Number of binary local rules.}
\label{tab:irr}
\begin{tabular}{|l|c|c|c|c|c|c|} \hline
$n$       & 0 & 1 & 2  & 3   & 4       & 5 \\ \hline
$|\overline{L}_2(n)|$   & 2 & 2 & 10 & 218 & 64\,594 & 4\,294\,642\,034\\ \hline
$|L_2(n)|$              & 2 & 4 & 16 & 256 & 65\,536 & 4\,294\,967\,296\\ \hline
\end{tabular}
\end{center}
\end{table}

The cardinalities of the binary local rule spaces for a small number of arguments are
depicted in Table~\ref{tab:irr}.
Let $H_\mathcal{A}(n, m)$ denote the set of projections $h : \mathcal{A}^n \to \mathcal{A}^m$.
The definition of an irreducible function implies the first part of the following proposition,
which takes advantage of the fact that a function $f: \emptyset \to \mathcal{A}$ is constant.
The other two statements are simple conclusions from Proposition~\ref{prop:uniqueness}.
\begin{prop}
\label{prop:props}
Suppose that $n \ge 0$, and $N$ is a finite set of integers.
Then
\begin{enumerate}[label=(\roman*)]
\item
$
L_\mathcal{A}(n) = \bigcup_{0 \le m \le n} \{ g \circ h \ | \ g \in \overline{L}_\mathcal{A}(m), h \in H_\mathcal{A}(n,m) \};
$
\item
$
G_\mathcal{A}(N) = \bigcup_{M\subseteq N} \overline{G}_\mathcal{A}(M);
$
\item
$
G_\mathcal{A}(N_1) \cap G_\mathcal{A}(N_2) = G_\mathcal{A}(N_1 \cap N_2).
$
\end{enumerate}
\end{prop}
\section{Local Rules}
\label{sec:loc}

\subsection{Symmetry Transformations}
\label{sec:loc-symm}

We now describe two transformations on the local rule space, one is based
on the permutation of the state set, the other one results from a reflection of a configuration.

Let $S_\mathcal{A}$ be the symmetric group of degree $|\mathcal{A}|$,
that is the set of all permutations of the set $\mathcal{A}$,
and suppose $\nu \in S_\mathcal{A}$.
If $a \in \mathcal{A}$ we write the image of $a$ under $\nu$ as product $\nu a$.
The extension of $\nu$ to words is defined by elementwise application.
If $w = a_1\ldots a_{n} \in \mathcal{A}^n$ is a word,
set $\nu w  = (\nu a_1) \ldots (\nu a_{n})$.
Suppose $f$ is a local rule from $\mathcal{A}^n$ to $\mathcal{A}$.
The permutation operator $\hat{\nu}$ is
defined by
$\hat{\nu}f(w) = \nu f ( \nu^{-1} w)$ for all words $w \in \mathcal{A}^n$.
Note that the  ``hat'' on the operator is necessary, because $\nu f$ and
$\hat{\nu} f$ are distinct entities.
The first one is the composite function $\nu \circ f$, whereas the second represents
the composite function $\nu \circ f \circ \nu^{-1}$.

If $w = a_1\ldots a_{n} \in \mathcal{A}^n$ is a word,
define the reflection of $w$ by $rw  = a_{n} \ldots a_1$.
Let $f$ be a local rule.
The reflection operator $\hat{r}$ is defined by $\hat{r} f(w) = f(rw)$.
Since $r$ is self-inverse ($r^{-1} = r$),
and $ra = a$ for all symbols $a \in \mathcal{A}$, we can also write $\hat{r} f(w) = rf(r^{-1}w)$,
which makes the notation consistent with that of the permutation operator.

Both operators, the permutation and the reflection operator, transform the rule
space $L_\mathcal{A}(n)$.
Because both of the operators map an irreducible local rule to another irreducible rule,
the operators can also be restricted to the rule space $\overline{L}_\mathcal{A}(n)$.
The set $R = \{1, r\}$, is called the reflection group.
The direct product of $R$ and $S_\mathcal{A}$, written as $RS_\mathcal{A}$, is the group that contains all permutations,
the reflection and their products.
We denote the group $RS_\mathcal{A}$ by  $\mathcal{S}_\mathcal{A}$.

Suppose that $\alpha$ and $\beta$ are in $\mathcal{S}_\mathcal{A}$.
Then
\[
\hat{\alpha} \hat{\beta} f (w)  = \alpha \hat{\beta} f(\alpha^{-1} w)
= \alpha \beta f( \beta^{-1} \alpha^{-1} w) = \widehat{\alpha \beta} f(w).
\]
The operators form a group that is in general isomorphic to $\mathcal{S}_\mathcal{A}$,
but for $n=1$ (or $k=1$) the relation is only a homomorphism.
Note that the reflection operator commutes with all permutation operators.
If $\alpha \in \mathcal{S}_\mathcal{A}$, and $f$ is a local rule such that $\hat{\alpha} f = f$,
the local rule $f$ is said to be invariant under the symmetry transformation $\alpha$.

A group action of a group $G$ on a set $A$ is a map of $G \times A$ into $A$
satisfying the following properties, see e.g.~\cite{rotman2012introduction}:
\begin{enumerate}[label=(\roman*)]
\item $g_1(g_2a) = (g_1g_2)a$ for all $g_1,g_2 \in G$, $a \in A$, and
\item $1a=a$, for all $a \in A$.
\end{enumerate}

The map $\mathcal{S}_\mathcal{A} \times L_\mathcal{A}(n) \to L_\mathcal{A}(n)$; $(\hat{\alpha}, f) \mapsto \hat{\alpha} f$
fulfils the properties of a group action.
The relation on $L_\mathcal{A}(n)$, defined by $f \sim g$ if and only if $f = \hat{\alpha} g$ for some
$\alpha \in \mathcal{S}_\mathcal{A}$, is an equivalence relation.
The equivalence classes $[f] = \{ \hat{\alpha}f \, | \, \alpha \in \mathcal{S}_\mathcal{A} \}$ form a partition
of the rule space $L_\mathcal{A}(n)$.
The set of all equivalence classes of $L_\mathcal{A}(n)$ is denoted by $L_\mathcal{A}(n) / \mathcal{S}_\mathcal{A}$.

Let $U$ be a subgroup of $\mathcal{S}_\mathcal{A}$ (denoted $U \le \mathcal{S}_\mathcal{A}$).
The conjugacy class of $U$ is the set of subgroups
$[U] = \{ V \le \mathcal{S}_\mathcal{A} \, | \, \mbox{ $V = \alpha U \alpha^{-1}$ for some $\alpha \in \mathcal{S}_\mathcal{A}$} \}$.
The stabilizer of a local rule $f \in L_\mathcal{A}(n)$ is the set
$\stab(f) = \{ \alpha \in \mathcal{S}_\mathcal{A} \ | \alpha f = f \}$.
An equivalence class $[f] \in L_\mathcal{A} / \mathcal{S}_\mathcal{A}$ is said to be of
type $[U]$ if the stabilizer of some $g$ in $[f]$ belongs to $[U]$.

If $\mathcal{A} = \{0,1\}$, we write $S_2$ for $S_{\{0,1\}}$, $L_2(n)$ for $L_{\{0,1\}}(n)$, and so on for the other notation
we introduced so far.
The group $S_2$ consists of the two elements $1$ and $c$, where $1$ is the identity, and $c = (01)$ is the transposition of
0 and 1.
The group $\mathcal{S}_2$ is Abelian, consists of the four elements $1$, $r$, $c$, and $rc$, and is isomorphic to the Klein four-group.
For an Abelian group, the conjugacy classes of subgroups are singletons, and we can identify the type of the group with
the group itself.
Hence the types of $\mathcal{S}_2$ are the subgroups of $\mathcal{S}_2$ which are $\langle 1 \rangle$, $\langle r \rangle$, $\langle c \rangle$,
$\langle rc \rangle$, and $\langle r, c \rangle = \mathcal{S}_2$.
The notation $\langle \alpha, \beta, \ldots \rangle$ is called a generator, and denotes the group with
the property that every element of the group can be written as finite
product of the elements of the generator and their inverses.
Define
\[
t^U_n = \{ [f] \in L_2(n) / \mathcal{S}_2 \, | \, \stab(f) = U \}
\]
for $U \le \mathcal{S}_2$.
The set $t^U_n$ consists of all equivalence classes of $L_2(n)$ that have type $U$.
The sets $t^U_n$, $U \le \mathcal{S}_2$, form a partition of $L_2(n) / \mathcal{S}_\mathcal{A}$, so
\[
L_2(n) / \mathcal{S}_2 = \bigcup_{U \le \mathcal{S}_2} t^U_n.
\]
The calculation of the number of equivalence classes
$|L_2(n) / \mathcal{S}_2|$ and of the numbers $|t^U_n|$, $U \le \mathcal{S}_2$,
can be found in \cite{SCHALLER2025100298}, where groups, group actions, and related concepts in the context of CAs
are also presented in more detail.

\subsection{Irreducible Binary Rules}

The restriction of the group action to
$\mathcal{S}_2 \times \overline{L}_2(n) \to \overline{L}_2(n)$ is well-defined.
The group action leads to the equivalence classes $\overline{L}_2(n) / \mathcal{S}_2$
and the classification by types
\[
\overline{t}^U_n = \{ [f] \in \overline{L}_2(n) / \mathcal{S}_2 \, | \, \stab(f) = U \}
\]
 for $U \le \mathcal{S}_2$.
We will now derive the number of equivalence classes of $\overline{L}_2(n) / \mathcal{S}_2$ and the numbers
$|\overline{t}^U_n|$, $U \le \mathcal{S}_2$,
from the numbers $|t^U_n|$, $U \le \mathcal{S}_2$.

For any projection $h$, we define $\hat{r} h = r \circ h \circ r$.
The projection $h$ is said to be reflection-invariant if $\hat{r}h = h$.
Let $g \circ h$ be a local rule,  where $g$ is an irreducible local rule and $h$ is a projection.
Then
\[
\hat{r}(g \circ h)  = g \circ h \circ  r  = g \circ r \circ r \circ h \circ r =
\hat{r}g \circ \hat{r}h.
\]
We have adopted the convention that a unary function operator binds more strongly than function composition.
It follows that if $f = g \circ h$, then $\hat{r} f = f$ if and only if $g = \hat{r}g$ and $h = \hat{r}h$.
Suppose $h: \mathcal{A}^n \to \mathcal{A}^m$ is a projection and $H(h) = \{p_1, \ldots, p_m\}$.
Then
\begin{align*}
\hat{r}h(a_1\ldots a_n) & = r \circ h \circ r(a_1\ldots a_n) =
r \circ h(a_n \ldots a_1) \\
& = r(a_{n+1-{p_1}} \ldots a_{n+1-{p_m}})
= a_{n+1-{p_m}} \ldots a_{n+1-p_1}.
\end{align*}
We see that the index set of $\hat{r}h$ relates to the index set of $h$ by
$H(\hat{r}h) = n + 1 - H(h)$.
We need to know how many projections $h: \mathcal{A}^n \to \mathcal{A}^m$ are reflection-invariant.
The question can be rephrased by asking how many ways there are to arrange $m$ balls symmetrically into
$n$ slots, which each slot holding at most one ball.
\begin{prop}
\label{prop:inv-h}
Suppose there are $m$ balls and $n$ slots where $n \ge m$ and each slot can hold one ball.
Then there are $\alpha_{n,m}$ arrangements of balls in slots that are reflection-symmetrical,
where $\alpha_{n,m}$ is defined by
\[
\alpha_{n,m} = \left\{
\begin{array}{ll}
0 & \mbox{if $n$ is even and $m$ is odd}, \\
\binom{\lfloor \frac{n}{2} \rfloor }{\lfloor \frac{m}{2} \rfloor } & \mbox{otherwise;}
\end{array}
\right.
\]
and where $\lfloor.\rfloor : \mathbb{R} \to \mathbb{Z}$,
$\lfloor x \rfloor = \mathrm{max}\{ p \in \mathbb{Z} \, | \, p \le x\}$, is the floor function.
The number of arrangements that are not reflection-symmetrical is denoted by $\beta_{n,m}$ defined by
\[
\beta_{n,m} = \binom{n}{m} - \alpha_{n,m}.
\]
\end{prop}
\begin{proof}
We distinguish between four cases in which $m$ is even or odd and $n$ is even or odd.

If $m$ and $n$ are even, there are $\binom{n/2}{m/2}$ ways to
arrange $m/2$ balls on the left half of the slots.
The arrangement of the balls on the right half of the slots follows from the reflection symmetry.

If $m$ is even and $n$ is odd, the slot in the center must remain empty,
so there are $\binom{(n-1)/2}{m/2}$ possibilities.

If $m$ is odd and $n$ is odd, and if $m \ge 1$, one ball must go into the slot in the middle.
Therefore, there are $\binom{(n-1)/2}{(m-1)/2}$ possibilities.

If $m$ is odd and $n$ is even, the number of balls in the left half of the slots cannot be equal
to the number of balls in the right half of the slots.
Therefore, there is no arrangement.

This proves the statement about $\alpha_{n,m}$.
Since $\alpha_{n,m} + \beta_{n,m} = \binom{n}{m}$ the statement about $\beta_{n,m}$ follows.
\end{proof}
Note that $\alpha_{n,n} = 1$ and $\beta_{n,n} = 0$.
We will now determine how the numbers $|t^U_n|$ can be expressed as a combination of
numbers $|\overline{t}^{V}_m|$.
We do this by analysing how an equivalence class $[g] \in \overline{t}^{V}_m$, $0 \le m \le n$, contributes
to $t^U_n$.
We structure this approach by iterating the possible types of equivalence classes, i.e., the subgroups of $\mathcal{S}_2$.
Let $0 \le m \le n$.

\begin{enumerate}

\item Type $\langle 1 \rangle$.
Let $[g] \in \overline{t}^{\langle 1 \rangle}_m$.
The equivalence class $[g]$ consists of the distinct local rules
$ \{ g, \hat{c}g, \hat{r}g, \widehat{rc}g \}$.
We form the compositions
of the local rules in $[g]$ and a projection $h: \mathcal{A}^n \to \mathcal{A}^m$
as well as the compositions of the local rules in $[g]$ and the projection $\hat{r}h$.
Put $A = [g] \circ h = \{g \circ h, \hat{c}g \circ h, \hat{r}g \circ h, \widehat{rc}g \circ h \}$,
and $B = [g] \circ \hat{r}h = \{g \circ \hat{r}h, \hat{c}g \circ \hat{r}h,
\hat{r}g \circ \hat{r}h, \widehat{rc}g \circ \hat{r}h \}$.
If $h = \hat{r}h$, then $A = B$ and
$A = [g \circ h, \hat{c}(g \circ h), \hat{r}(g \circ h), \widehat{rc}(g \circ h)]$.
Thus $A$ is an equivalence class of $t^{\langle 1 \rangle}_n$.
Suppose now that $h \ne \hat{r}h$.
The sets $A$ and $B$ are not equivalence classes of $t^{\langle 1 \rangle}_n$, but
$A \cup B$ can be partitioned into two equivalence classes:
$[g \circ h] = \{g \circ h, \hat{c}g \circ h,
\hat{r} g \circ \hat{r} h, \widehat{rc} g \circ \hat{r} h \}$,
and
$[g \circ \hat{r}h] = \{g \circ \hat{r}h, \hat{c}g \circ \hat{r}h,
\hat{r} g \circ h, \widehat{rc} g \circ h \}$.
So, either $[g]$ and $h = \hat{r}h$ contribute one equivalence class to $t^{\langle 1 \rangle}_n$,
or $[g]$, $h$, and $\hat{r}h$, ($h \ne \hat{r}h$), contribute two
equivalence classes to $t^{\langle 1 \rangle}_n$.
Either way, we conclude that each pair $[g], h$ induces one equivalence class in $t^{\langle 1 \rangle}_n$.
Thus $\overline{t}^{\langle 1 \rangle}_m$
contributes  $\binom{n}{m} |\overline{t}^{\langle 1 \rangle}_m|$ equivalence classes to $t^{\langle 1 \rangle}_n$.

\item Type $\langle c \rangle$.
Let $[g] \in \overline{t}^{\langle c \rangle}_m$.
The equivalence class $[g]$ consists of the distinct local rules
$ \{ g, \hat{r}g \}$.
Put $A = [g] \circ h = \{g \circ h, \hat{r}g \circ h \}$, and
$B = [g] \circ \hat{r}h = \{g \circ \hat{r}h, \hat{r}g \circ \hat{r}h \}$.
If $h = \hat{r}h$, then $A = B$ and $A$ is an equivalence class of $t^{\langle c \rangle}_n$.
If $h \ne \hat{r}h$,
then $A \cup B$ can be decomposed into the two equivalence classes
$[g \circ h] = \{g \circ h, \hat{r}g \circ \hat{r}h\}$
and
$[g \circ \hat{r}h] = \{g \circ \hat{r}h, \hat{r}g \circ h\}$.
Similar to the case above, $\overline{t}^{\langle c \rangle}_m$
contributes  $\binom{n}{m} |\overline{t}^{\langle c \rangle}_m|$ equivalence classes to $t^{\langle c \rangle}_n$.

\item Type $\langle r \rangle$.
Let $[g] \in \overline{t}^{\langle r \rangle}_m$.
The equivalence class $[g]$ consists of the distinct local rules
$ \{ g, \hat{c}g \}$.
Put $A = [g] \circ h = \{g \circ h, \hat{c}g \circ h \}$, and
$B = [g] \circ \hat{r}h = \{g \circ \hat{r}h, \hat{c}g \circ \hat{r}h \}$.
If $h = \hat{r}h$, then $A = B$, and $A \in t^{\langle r \rangle}_n$.
If $h \ne \hat{r}h$, then $A \cup B$ form one equivalence class of $t^{\langle 1\rangle}_n$.
Thus $|\overline{t}^{\langle r \rangle}_m|$
contributes $\alpha_{n,m} |\overline{t}^{\langle r \rangle}_m|$ equivalence classes to $t^{\langle r \rangle}_n$
and $\beta_{n,m}|\overline{t}^{\langle r \rangle}_m| / 2$ equivalence classes to $t^{\langle 1\rangle}_n$.

\item Type $\langle rc \rangle$.
Let $[g] \in \overline{t}^{\langle rc \rangle}_m$.
The equivalence class $[g]$ consists of the distinct local rules
$ \{ g, \hat{r}g \}$.
Note that $\hat{c}g = \hat{r}g$.
Put $A = [g] \circ h = \{g \circ h, \hat{r}g \circ h \}$, and
$B = [g] \circ \hat{r}h = \{g \circ \hat{r}h, \hat{r}g \circ \hat{r}h \}$.
If $h = \hat{r}h$, then $A = B$, and $A \in t^{\langle rc \rangle}_n$.
If $h \ne \hat{r}h$, then $A \cup B$ form one equivalence class of $t^{\langle 1\rangle}_n$.
Thus $\overline{t}^{\langle rc \rangle}_m$
contributes $\alpha_{n,m} |\overline{t}^{\langle rc \rangle}_m| $ equivalence classes to $t^{\langle rc \rangle}_n$
and $\beta_{n,m} |\overline{t}^{\langle rc \rangle}_m| / 2$ equivalence classes to $t^{\langle 1\rangle}_n$.

\item Type $\mathcal{S}_2$.
Let $[g] \in \overline{t}^{\mathcal{S}_2}_m$.
Then $[g] = \{ g \}$.
Put $A = [g] \circ h = \{g \circ h \}$, and
$B = [g] \circ \hat{r}h = \{g \circ \hat{r}h  \}$.
If $h = \hat{r}h$, then $A = B$, and $A \in t^{\mathcal{S}_2}_n$.
If $h \ne \hat{r}h$, then $A \cup B$ form one equivalence class of $t^{\langle c \rangle}_n$.
Thus $\overline{t}^{\mathcal{S}_2}_m$
contributes $\alpha_{n,m} |\overline{t}^{\mathcal{S}_2}_m|$ equivalence classes to $t^{\mathcal{S}_2}_n$
and $\beta_{n,m}  |\overline{t}^{\mathcal{S}_2}_m| / 2$ equivalence classes to $t^{\langle c \rangle}_n$.
\end{enumerate}
Collecting all contributions to $t^{\langle 1 \rangle}_n$, we obtain
\[
|t_n^{\langle 1 \rangle}| =
         \sum_{m=0}^{n} \left[ \binom{n}{m} |\overline{t}_m^{\langle 1 \rangle}|
        + \frac{1}{2} \beta_{n,m}  \left( |\overline{t}_m^{\langle r \rangle}|
        + |\overline{t}_m^{\langle rc \rangle}| \right) \right].
\]
Since our goal is to calculate $|\overline{t}^{\langle 1 \rangle}_n|$ we
rewrite the equation using $\alpha_{n,n} = 1$ and $\beta_{n,n} = 0$:
\[
|\overline{t}_n^{\langle 1 \rangle}| =  |t_n^{\langle 1 \rangle}|
        - \sum_{m=0}^{n-1} \left[ \binom{n}{m} |\overline{t}_m^{\langle 1 \rangle}|
        + \frac{1}{2} \beta_{n,m}  \left( |\overline{t}_m^{\langle r \rangle}|
        + |\overline{t}_m^{\langle rc \rangle}| \right) \right].
\]
After considering the other cases, and rewriting the corresponding equations, we obtain the following proposition.
\begin{prop}
The cardinalities of the equivalence classes of $\overline{L}(n,2)$ by type are given by
\begin{equation*}
  \label{eq:t}
  \begin{aligned}
|\overline{t}_n^{\mathcal{S}_2}| &= |t_n^{\mathcal{S}_2}| - \sum_{m=0}^{n-1} \alpha_{n,m} |\overline{t}_m^{\mathcal{S}_2}|; \\
|\overline{t}_n^{\langle r \rangle}| &= |t_n^{\langle r \rangle}|
        - \sum_{m=0}^{n-1} \alpha_{n,m} |\overline{t}_m^{\langle r \rangle}|; \\
|\overline{t}_n^{\langle rc \rangle}| &= |t_n^{\langle rc \rangle}|
        - \sum_{m=0}^{n-1} \alpha_{n,m} |\overline{t}_m^{\langle rc \rangle}|; \\
|\overline{t}_n^{\langle c \rangle}| &= |t_n^{\langle c \rangle}|
        - \sum_{m=0}^{n-1} \left[ \binom{n}{m} |\overline{t}_m^{\langle c \rangle}|
        + \frac{1}{2} \beta_{n,m} |\overline{t}_m^{\mathcal{S}_2}| \right]; \\
|\overline{t}_n^{\langle 1 \rangle}| &=  |t_n^{\langle 1 \rangle}|
        - \sum_{m=0}^{n-1} \left[ \binom{n}{m} |\overline{t}_m^{\langle 1 \rangle}|
        + \frac{1}{2} \beta_{n,m}  \left( |\overline{t}_m^{\langle r \rangle}|
        + |\overline{t}_m^{\langle rc \rangle}| \right) \right].
  \end{aligned}
\end{equation*}
\end{prop}

\begin{table}
\begin{center}
\caption{Number of equivalence classes of irreducible local rules.}
\label{tab:s2_fs}
\begin{tabular}{|l|r|r|r|r|r|r|}
\hline
       & \multicolumn{6}{|c|}{$|\overline{t}^H_n|$} \\ \hline
$H$    & $n=0$ & $n=1$ & $n=2$ & $n=3$ & $n=4$ & $n=5$ \\ \hline
$\langle S_2 \rangle$                   & 0 & 2 & 0 &  6 &       0 & 1\,010 \\ \hline
$\langle rc \rangle$                            & 0 & 0 & 0 &  4 &       0 & 32\,248 \\ \hline
$\langle c \rangle$                             & 0 & 0 & 0 &  2 &     104 & 31\,668 \\ \hline
$\langle r \rangle$                             & 1 & 0 & 3 & 24 &     505 & 523\,216 \\ \hline
$\langle 1 \rangle$                             & 0 & 0 & 1 & 38 & 15\,844 & 1\,073\,336\,680 \\ \hline
$|\overline{L}_2(n)/\mathcal{S}_2|$     & 1 & 2 & 4 & 74 & 16\,453 & 1\,073\,954\,842 \\ \hline
\end{tabular}
\end{center}
\end{table}

The numbers $|\overline{t}^H_n|$ can now iteratively calculated, starting by $n=0$.
Table~\ref{tab:s2_fs} depicts the results for a neighbourhood size $n=0, \ldots, 5$.
The sets $\overline{t}^U_n$ form a partition of $\overline{L}_2(n) / \mathcal{S}_2$, so
\[
|\overline{L}_2(n) / \mathcal{S}_2| = \sum_{U \le \mathcal{S}_2} |\overline{t}^U_n|.
\]

\subsection{Small Neighbourhoods}

We still limit ourselves to binary CAs.
There are only seven equivalence classes for a neighbourhood size less than or equal to two, which we
list below.

\begin{enumerate}

\item Cardinality Zero, $\overline{L}_2(0)$.
These are the constant functions $W^0_0 = 0$ and $W^0_1 = \hat{c} W^0_0 = 1$.

\item Cardinality One, $\overline{L}_2(1)$.
One function is the identity $W^1_2 = \id$, $W^1_2(a_0) = a_0$ for $a_0 \in \{0,1\}$.
The second function inverts the $\id$ function, $W^1_1(a_0) = ca_0$, $W^1_1 = c \id$.
Note that $\hat{c} \id = c \circ \id \circ c = \id \ne c \circ \id$.

\item Cardinality Two, $\overline{L}_2(2)$.

Equivalence Class $\{W^2_1, W^2_7\}$.
The two functions are $W^2_1$, given by
\[
W^2_1(a_0a_1) =
\left\{
\begin{array}{ll}
1 & \mbox{if $a_0a_1 = 00$;} \\
0 & \mbox{otherwise,}
\end{array}
\right.
\]
and $W^2_{7} = \hat{c}W^2_1$.

Equivalence Class $\{W^2_6, W^2_9\}$.
The two functions are
$W^2_6$, given by
\[
W^2_6(a_0a_1) =
\left\{
\begin{array}{ll}
1 & \mbox{if $a_0a_1 = 01$ or $a_0a_1 = 10$;} \\
0 & \mbox{otherwise,}
\end{array}
\right.
\]
and $W^2_9 = \hat{c}W^2_6$.

Equivalence Class $\{W^2_8, W^2_{14}\}$.
The two functions are
$W^2_8$, given by
\[
W^2_8(a_0a_1) =
\left\{
\begin{array}{ll}
1 & \mbox{if $a_0a_1 = 11$;} \\
0 & \mbox{otherwise,}
\end{array}
\right.
\]
and $W^2_{14} = \hat{c}W^2_8$.

Equivalence Class
$\{W^2_2,W^2_4,W^2_{11},W^2_{13}\}$.
The first function is given by
\[
W^2_2(a_0a_1) =
\left\{
\begin{array}{ll}
1 & \mbox{if $a_0a_1 = 01$;} \\
0 & \mbox{otherwise,}
\end{array}
\right.
\]
The other ones are $W^2_4 = \hat{r}W^2_2$, $W^2_{11} = \hat{c}W^2_2$, and
$W^2_{13} = \widehat{rc} W^2_2$.

\end{enumerate}

\section{Global Maps}
\label{sec:glob}

\subsection{Symmetry Transformations}
We now examine the symmetry transformations on global maps and link two of them to the
symmetry transformations on local rules, see  Subsection \ref{sec:loc-symm}.

First we consider the permutation operator.
Let $S_\mathcal{A}$ be the symmetric group of degree $|\mathcal{A}|$
and suppose $\nu \in S_\mathcal{A}$.
The extension of $\nu$ to configurations is defined by elementwise application.
If $x$ is a configuration, set $(\nu x)_i = \nu (x_i)$.
If $\Phi^N_f$ is the induced global mapping of $f$, we define
$\hat{\nu} \Phi^N_f$ by
$\hat{\nu} \Phi^N_f (x) = \nu \Phi^N_f (\nu^{-1} x)$ for all configurations $x$.
From
\begin{eqnarray*}
\left( \hat{\nu} \Phi^N_f(x) \right)_{i}
& = & \left( \nu \Phi^N_f (\nu^{-1} x) \right)_{i}
= \nu f \left(\nu^{-1}(x_{i+j_1} \ldots x_{i+j_{n}}) \right) \\
& = & \hat{\nu} f(x_{i+j_1} \ldots x_{i+j_{n}})
 =  \Phi^N_{\hat{\nu} f} (x)_{i}
\end{eqnarray*}
follows $\hat{\nu} \Phi^N_f =  \Phi^N_{\hat{\nu} f}$.
Let $A = (\mathcal{A}, N, f)$ be a CA and put $B = (\mathcal{A}, N, \hat{\nu} f)$.
By defining $\hat{\nu} A = B$, the permutation operator $\hat{\nu}$ maps the CA $A$ to the CA $B$.

The two other transformations, reflection and shift, originate from isometric transformations on the set of integers.
Isometric transformations, that is transformations that preserve the distance, form a group,
the infinite dihedral group $D_\infty$, see e.g. \cite{armstrong1997groups}.
By definition if $\varphi \in D_\infty$ and $i,j \in \mathbb{Z}$,
then $|\varphi(i) - \varphi(j)| = |i - j|$.
The group $D_\infty$ has the representation $\langle r, t \, | \, r^2=1, rtr = t^{-1} \rangle$,
where $r$ and $t$ can be chosen to be $r(i) = -i$ and $t(i) = i+1$ for all $i \in \mathbb{Z}$.

We extend these operations to sets of integers.
If $N$ is any set of integers, put $-N = \{-i \, | \, i \in N \}$ and
$j + N = \{ j + i \, | \, i \in N \}$ if $j$ is itself an integer.
Two sets of integers $M$ and $N$ are called congruent, in symbols $M \cong N$,  if and only if
there is an integer $j$ such that $N = j + M $ or $N = j - M$.
The congruence relation is an equivalence relation.
A set $N$ of integers is called reflection-symmetrical, if
there is an integer $j$ such that $N = j - N$.
Let $\Sym$ denote the set of all reflection-symmetrical subsets of integers.
Then $N$ is reflection-symmetrical if and only if $N \in \Sym$.
If $N \in \Sym$, then $[N] \subset \Sym$.

The shift operator $\sigma$ plays a dual role.
First, by the Curtis-Hedlund-Lyndon Theorem the shift operator characterizes the continuous functions on $\mathcal{A}^\mathbb{Z}$
that can be expressed as the global maps of a CA.
Second it can be used to map a CA to another one.
It is this second meaning that we now study.
Let $A=(\mathcal{A},N,f)$ be a CA.
Then
\[
(\sigma \Phi^N_f(x))_i = (\Phi^N_f(x))_{i+1}
= f(x_{i+j_1+1} \ldots x_{i+j_{n}+1}) = ( \Phi^{1+N}_f(x))_i.
\]
Hence we also can see $\sigma$ as the operator that maps the CA $A = (\mathcal{A},N,f)$ to
the CA $B = \sigma A = (\mathcal{A}, 1+N,f)$,
expressed by
\[
\sigma \Phi^A = \Phi^{\sigma A} \ \ \ \mbox{or} \ \ \ \sigma \Phi^N_f = \Phi^{1+N}_f.
\]

The reflection operator is based on the reflections of configurations.
If $x$ is a configuration, set $(rx)_i = x_{-i}$.
The reflection operator $\hat{r}$ is formally defined by  $\hat{r} \Phi_f(x) = r \Phi_f(rx)$.
Let $A = (\mathcal{A},N,f)$ be a CA with the global mapping $\Phi^N_f$.
Then
\begin{align*}
(\hat{r} \Phi^N_f(x))_i & =  (r\Phi^N_f(rx))_i
        =  (\Phi^N_f(rx))_{-i} = f((rx)_{-i+j_1} \ldots (rx)_{-i+j_n}) \\
& =  f(x_{i-j_1} \ldots x_{i-j_n}) = f(r(x_{i-j_n} \ldots x_{i-j_1})) \\
& =  \hat{r} f(x_{i-j_n} \ldots x_{i-j_1}) = (\Phi^{-N}_{\hat{r}f}(x))_i.
\end{align*}
If we define $\hat{r}A = (\mathcal{A}, -N, \hat{r}f)$, we can express the relation above as
$\hat{r} \Phi^A = \Phi^{\hat{r} A}$.
Let $x$ be a configuration.
From
\[
(\sigma r (x))_i = (r(x))_{i+1} = x_{-i-1} = (\sigma^{-1}(x))_{-i} = (r \sigma^{-1}(x))_i
\]
follows the commutator relation $\sigma r = \sigma^{-1}r$.
Then
\[
\sigma \hat{r} \Phi^N_f = \sigma r \Phi^N_f r = r \sigma^{-1} \Phi^N_f r = \hat{r} \sigma^{-1} \Phi^N_f
\]
also implies $\sigma \hat{r} = \hat{r} \sigma^{-1}$.

In general, the reflection operator is not a transformation on the global rule space $G_\mathcal{A}(N)$ or
$\overline{G}_\mathcal{A}(N)$,
but a map $\hat{r}: G_\mathcal{A}(N) \to G_\mathcal{A}(-N)$ or $\hat{r}: \overline{G}_\mathcal{A}(N) \to
\overline{G}_\mathcal{A}(-N)$.
If $N = -N$, the map becomes a transformation.

The operators $\hat{v}$, $\sigma$, and $\hat{r}$ that act on a global map are defined in terms of a local rule.
Since different local rules can induce the same global map, we have to show that
the operators are well-defined.

\begin{prop}
Let $A = (\mathcal{A}, N_1, f_1)$ and $B = (\mathcal{A}, N_2, f_2)$ be two CAs.
If $\Phi^A = \Phi^B$, then
$\sigma \Phi^A = \sigma \Phi^B$, $\hat{\nu} \Phi^A = \hat{\nu} \Phi^B$ if $\nu \in S_\mathcal{A}$,
and $\hat{r} \Phi^A = \hat{r} \Phi^B$.
\end{prop}
\begin{proof}
Let $f$ be a reducible local rule that can be written as $f = g \circ h$, where $g$ is an irreducible
local rule and $h$ is a projection.

If $N = \{j_1, \ldots, j_n\}$ is a set of integers
such that $j_1 < \ldots < j_n$ and
and $H = \{p_1, \ldots, p_m\}$ is a subset of the integer interval $[1,n]$
such that $p_1 < \ldots < p_m$,
we define the selection operator $|$ by
$N | H = \{j_{p_1}, \ldots, j_{p_m}\}$.
Then
\begin{align*}
\Phi^N_f(x)_i & = (g \circ h(x))_i = g \circ h(x_{i+j_1} \ldots x_{i+j_n})
= g(x_{i+j_{p_1}} \ldots x_{i+j_{p_m}}) = \Phi^{N | H(h)}_g.
\end{align*}
The representation $\Phi^N_f = \Phi^{N | H(h)}_g$ of $\Phi^N_f$ by an irreducible local rule $g$
and the neighbourhood $N | H(h)$ is unique.
If $\Phi^N_f = \Phi^M_{g^\prime}$ and $g^\prime$ is irreducible,
it follows from Proposition~\ref{prop:uniqueness} that $g=g^\prime$ and
$M = N | H(h)$.
Hence it is sufficient to show that $\Phi^N_f$ and $\Phi^{N | H(h)}_g$ have the same image under
the operators $\hat{\nu}$, $\sigma$, and $\hat{r}$.
From $\hat{\nu}f = \hat{\nu}g \circ h$ follows
$\hat{\nu} \Phi^N_f = \hat{\nu} \Phi^{N | H(h)}_g$.
Also $\sigma \Phi^N_f = \Phi^{1+N}_f = \Phi^{(1 + N )| H(h))}_g =
\Phi^{1 + (N | H(h))}_g = \sigma \Phi^{N | H(h)}_g$
is easily verified.
The relation
\begin{align*}
(-N) | H(\hat{r}h) & = \{-j_n, \ldots, -j_1 \} | (n+1-H(h)) \\
& = \{-j_n, \ldots, -j_1 \} | \{ n+1-j_{p_m}, \ldots, n+1-j_{p_1} \} \\
& = \{ -j_{p_m}, \ldots, -j_{p_{1}} \} = -(N | H(h)),
\end{align*}
implies $(-N) | H(\hat{r}h) = -(N | H(h))$.
Thus
\begin{align*}
\hat{r}\Phi^N_f = \Phi^{-N}_{\hat{r}f} = \Phi^{-N}_{\hat{r} g \circ \hat{r} h}
= \Phi^{(-N)|H(\hat{r}h)}_{\hat{r}g} = \Phi^{-(N|H(h))}_{\hat{r}g} =
\hat{r} \Phi^{N|H(h)}_g.
\end{align*}
\end{proof}
The products of all symmetry transformations form the infinite group
\[
\mathcal{T}_\mathcal{A} = \{ \sigma^i \hat{r}^j \hat{\nu} \, | \, i \in \mathbb{Z}; j=0,1; \nu \in S_\mathcal{A} \}.
\]
It acts as group on $\mathcal{C}_\mathcal{A}$, the set of all global maps induced by CAs over the
alphabet $\mathcal{A}$.
The group action partitions $\mathcal{C}_\mathcal{A}$ into the
equivalence classes $\mathcal{C}_\mathcal{A} / \mathcal{T}_\mathcal{A}$.
Suppose the global maps $\Phi^A$ and $\Phi^B$ are in $\mathcal{C}_\mathcal{A}$.
We write $\Phi^A \cong \Phi^B$ if and only if there is an equivalence class
$C \in  \mathcal{C}_\mathcal{A} / \mathcal{T}_\mathcal{A}$ such that
both $\Phi^A$ and $\Phi^B$ are members of $C$.
Being an equivalence relation, the restriction of the relation $\cong$ to
any subset $\mathcal{D}$ of $\mathcal{C}_\mathcal{A}$ is still an equivalence relation.
We then write $\mathcal{D} / {\cong}$ for the resulting equivalence classes.
Before we study the general case, we examine the equivalence classes of global functions induced by
a neighbourhood of size $n=0,1,2$.

\begin{enumerate}

\item Let $n=0$. Then $N=\emptyset$.
The only local maps with an empty domain are the constant functions.
If $f$ is a constant local map, then the equivalence class of $\Phi_f$ is the set
of all constant local maps on $\mathcal{A}^\mathbb{Z}$.
Thus $[\Phi^\emptyset_f] = \{ \Phi^\emptyset_g \, | \, g() \in \mathcal{A} \}$.
There is one equivalence class and the size of this class is $k$.

\item Let $n=1$, $N=[0]$, and $f \in \overline{L}_\mathcal{A}(1)$.
Then $[\Phi^{[0]}_f] = \{ \Phi^{[i]}_g \, | \, i \in \mathbb{Z}, g \in [f] \}$.
There are $|L_\mathcal{A}(1) / \mathcal{S}|$ equivalence classes of the form $[\Phi^{[0]}_f]$,
and each equivalence class is countably infinite.

\item Let $n=2$, $N=\{0,p\}$, and $f \in \overline{L}_\mathcal{A}(2)$.
From
$\sigma^{2p}\hat{r}\Phi^{\{0,p\}}_f = \sigma^{2p}\Phi^{\{-p,0\}}_{\hat{r}f}
= \Phi^{\{0,p\}}_{\hat{r}f}$
follows $\Phi^{\{0,p\}}_{\hat{r}f} \in [\Phi^{\{0,p\}}_f]$.
We conclude that $[\Phi^{\{0,p\}}_f] = \{ \Phi^{\{i,i+p\}}_g \, | \, i \in \mathbb{Z}, g \in [f] \}$.
In this case, there are countably infinitely many equivalence classes, and the size of each
equivalence class is countably infinite.

\end{enumerate}

We consider now the general case.
Let $N$ be a finite set of integers and $f \in \overline{L}_\mathcal{A}(|N|)$.

\begin{enumerate}

\item Suppose $N \in \Sym$.
Then $N = p -N$ for a integer $p$, and consequently $\Phi^{N}_{\hat{r}f} \in [\Phi^{N}_f]$.
Thus $[\Phi^N_f] = \{\Phi^{d+N}_g \, | \, d \in \mathbb{Z}, g \in [f] \}$.
Since $N \in \Sym$, the equivalence class can also be written as
$[\Phi^N_f] = \{\Phi^{M}_g \, | \, M \cong N, g \in [f] \}$.

\item Suppose $N \not\in \Sym$ and suppose there is a $\nu \in S_\mathcal{A}$ such
that $\widehat{r\nu}f = f$.
Then $\widehat{r\nu} \Phi^N_f = \Phi^{-N}_f$.
This implies as above that
$[\Phi^N_f] = \{\Phi^{M}_g \, | \, M \cong N, g \in [f] \}$.

\item Suppose $N \not\in \Sym$ and suppose there is no $\nu \in S_\mathcal{A}$
such that $\widehat{r\nu}f =f$.
Then $[\Phi^N_f] = \{ \Phi^{d+N}_{\hat{\nu}f},  \Phi^{d-N}_{\widehat{r\nu}f} \, | \, d \in \mathbb{Z},
\nu \in S_\mathcal{A} \}$.
Note that $\Phi^N_{\hat{r}f} \not\in [\Phi^N_f]$.

\end{enumerate}
By grouping the last two cases together, we obtain the following proposition.
\begin{prop}
\label{prop:global-eq}
Let $N$ be a finite set of integers and $f \in \overline{L}_\mathcal{A}(|N|)$.
\begin{enumerate}[label=(\roman*)]
\item
\label{prop:global-eq-1}
If $N \in \Sym$ or $\widehat{r\nu}f = f$ for some $\nu \in S_\mathcal{A}$ then
$N$ and $[f]$ induce the following equivalence class of $\mathcal{C}_\mathcal{A} / \mathcal{T}_\mathcal{A}$:
\[
[\Phi^N_f] = \{\Phi^{M}_g \, | \, M \cong N, g \in [f] \}.
\]
\item
\label{prop:global-eq-2}
Otherwise $N$ and $[f]$ induce the following two equivalence classes of $\mathcal{C}_\mathcal{A} / \mathcal{T}_\mathcal{A}$:
\begin{align*}
[\Phi^N_f] & = \{ \Phi^{d+N}_{\hat{\nu}f},  \Phi^{d-N}_{\widehat{r\nu}f} \, | \, d \in \mathbb{Z},
\nu \in S_\mathcal{A} \}; \\
[\Phi^N_{\hat{r}f}] & = \{ \Phi^{d+N}_{\widehat{r\nu}f},  \Phi^{d-N}_{\hat{\nu}f} \, | \, d \in \mathbb{Z},
\nu \in S_\mathcal{A} \}.
\end{align*}
\end{enumerate}
\end{prop}

\subsection{Finite Intervals}

This subsection deals with the equivalence classes of global maps or CAs whose neighbourhood
is a subset of a given finite integer interval.
We take a step back and first consider the equivalence relation based only on
permutations of the alphabet and reflection of the configuration.
Let $p$ be a nonnegative integer and consider the interval $N=[-p,p]$.
The map
\[
\mathcal{S}_\mathcal{A}  \times L_\mathcal{A}(2p+1) \to L_\mathcal{A}(2p+1); \
(\alpha, f) \mapsto \hat{\alpha}f
\]
is a group action, and so
is the map
\[
\mathcal{S}_\mathcal{A} \times G_\mathcal{A}(N) \to G_\mathcal{A}(N); \
(\alpha, \Phi^N_f) \mapsto \hat{\alpha} \Phi^N_f.
\]
We see that the map
\[
L_\mathcal{A}(2p+1) / \mathcal{S}_\mathcal{A} \to G_\mathcal{A}(N)/\mathcal{S}_\mathcal{A}; \
[f] \mapsto [\Phi^N_f]
\]
 is a bijection.
The same group acts both on the set of local maps and the set of global maps,
and the equivalence classes associated with the group actions are in one-to-one correspondence.

The size of $[-p,p]$ is always odd.
To deal with all neighbourhood sizes we now investigate the intervals $N = [1,n]$.
If $f \in L_\mathcal{A}(n)$ and $f$ is not constant, then
\[
\hat{r}\Phi^{N}_f = \Phi^{-N}_{\hat{r}f} \not \in G_\mathcal{A}(N).
\]
This means that we cannot use the same group action as above.
We note however that
\[
\sigma^{n+1}\hat{r}\Phi^N_f = \Phi^{N}_{\hat{r}f} \in G_\mathcal{A}(N).
\]
Let us therefore define the group
\[
\hat{\mathcal{S}}_\mathcal{A}(n+1) = \langle \sigma^{n+1}\hat{r}, \hat{\nu} \, | \, \nu \in S_\mathcal{A} \}.
\]
From $(\sigma^{n+1}\hat{r})^{-1} = \sigma^{n+1}\hat{r}$, we conclude that
$\hat{\mathcal{S}}_\mathcal{A}(n+1)$ is isomorphic to $\mathcal{S}_\mathcal{A}$.
It follows, that
\[
\hat{\mathcal{S}}_\mathcal{A}(n+1) \times G_\mathcal{A}(N) \to G_\mathcal{A}(N); \
(\hat{\alpha}, \Phi^N_f) \mapsto \hat{\alpha} \Phi^N_f
\]
is a group action and we conclude that
$L_\mathcal{A}(n) / \mathcal{S}_\mathcal{A}$
and $G_\mathcal{A}(N)/\hat{\mathcal{S}}_\mathcal{A}(n+1)$ can be again put in one-to-one correspondence.

Let us now include the shift operator and let us consider the group $\mathcal{T}_\mathcal{A}$
of symmetry operators.
Without loss of generality we now consider the intervals $[1,n]$;
the counting is identical for the intervals of the form $[p, p+n-1]$.
To start with, we present the following example.
We ask which equivalence classes of $G_2([1,4])/{\cong}$ are induced
by some of the local rules of $\overline{L}_2(3)$ and neighbourhoods that are subsets of the interval $[1,4]$.
There are  four 3-element subsets of the integer interval $[1,4]$, which we denote by
$A=[1,3]$, $B=[2,4]$, $C=\{1,2,4\}$, and $D=\{1,3,4\}$.
Let us now consider the equivalence classes
$[W_{54}]$, $[W_{57}]$, $[W_{110}] \in L_2(3) / \mathcal{S}_2$.
The types of these equivalence classes are $\langle r \rangle$, $\langle rc \rangle$, and $\langle 1 \rangle$, respectively.
Note that $\sigma^4 \hat{r} \Phi^A_f = \Phi^A_{\hat{r}f}$, $\sigma \Phi^A_f = \Phi^B_f$, and
$\sigma^5 \hat{r} \Phi^C_f = \Phi^D_{\hat{r}f}$.
We see that
\begin{align*}
[\Phi^{A}_{W_{54}}] & = \{ \Phi^{A}_{W_{54}}, \Phi^{A}_{\hat{c}W_{54}}, \Phi^{B}_{W_{54}},
        \Phi^{B}_{\hat{c}W_{54}} \}; \\
[\Phi^{A}_{W_{57}}] & = \{ \Phi^{A}_{W_{57}}, \Phi^{A}_{\hat{r}W_{54}}, \Phi^{B}_{W_{57}},
        \Phi^{B}_{\hat{r}W_{57}} \}; \\
[\Phi^{A}_{W_{110}}] & = \{ \Phi^{A}_{W_{110}}, \Phi^{A}_{\hat{c}W_{110}}, \Phi^{A}_{\hat{r}W_{110}},
        \Phi^{A}_{\widehat{rc}W_{110}},
        \Phi^{B}_{W_{110}}, \Phi^{B}_{\hat{c}W_{110}}, \Phi^{B}_{\hat{r}W_{110}},
        \Phi^{B}_{\widehat{rc}W_{110}}
\}; \\
[\Phi^{C}_{W_{54}}] & = \{ \Phi^{C}_{W_{54}}, \Phi^{C}_{\hat{c}W_{54}}, \Phi^{D}_{W_{54}},
        \Phi^{D}_{\hat{c}W_{54}} \}; \\
[\Phi^{C}_{W_{57}}] & = \{ \Phi^{C}_{W_{57}}, \Phi^{C}_{\hat{r}W_{57}}, \Phi^{D}_{W_{57}},
        \Phi^{D}_{\hat{r}W_{57}} \}; \\
[\Phi^{C}_{W_{110}}] & = \{ \Phi^{C}_{W_{110}}, \Phi^{C}_{\hat{c}W_{110}}, \Phi^{D}_{\hat{r}W_{110}},
        \Phi^{D}_{\widehat{rc}W_{110}} \}; \\
[\Phi^{C}_{\hat{r}W_{110}}] & = \{ \Phi^{C}_{\hat{r}W_{110}}, \Phi^{C}_{\widehat{rc}W_{110}},
        \Phi^{D}_{W_{110}}, \Phi^{D}_{\hat{c}W_{110}} \}.
\end{align*}
The neighbourhood $A$ and each of the equivalence classes $[W_{110}]$, $[W_{54}]$, and $[W_{57}]$
induce exactly one equivalence class of global maps in $G_2([1,4])/{\cong}$.
The situation is different for the neighbourhood $C$.
In this case, $C$ and each of $[W_{54}]$ and $[W_{57}]$ induce one equivalence class of global maps,
but $C$ and $[W_{110}]$ induce two.
We explain this by observing that $\hat{r}W_{110} \ne W_{110}$ and
$\hat{r}W_{110} \ne \hat{c}W_{110}$.
This example demonstrates that the equivalence classes of $G_2([1,4])/{\cong}$ can be determined by examining
the equivalence classes of $\overline{L}(n)/\mathcal{S}_2$, $0 \le n \le 4$, and the neighbourhoods contained
in $[1,4]$, while taking into account the types of these equivalence classes.

Earlier on, we defined the relation $\cong$ on the subsets of integers.
We restrict this relation now to the power set of a given subset $N$ of integers, and denote the equivalence
classes by $\mathcal{P}(N)/{\cong}$.
Thus $\mathcal{P}(N)/{\cong}$ is the set $ \{ [M] \, | \, M \subseteq N \}$, where $[M]$ is the equivalence
class $\{ M^\prime \subseteq N \, | \, M^\prime \cong M \}$.
We are using this notation in the following proposition.
\begin{prop}
\label{prop:global-shift}
Suppose $N = [1, n]$ is an integer interval. Then
\[
|G_2(N)/{\cong}| = \sum_{[M]\in \mathcal{P}(N)/{\cong}} |\overline{L}_2(|M|)/\mathcal{S}_2| +
      \sum_{[M]\in \mathcal{P}(N)/{\cong}, M \not\in Sym} (|t_{|M|}^{\langle 1 \rangle}| + (|t_{|M|}^{\langle c \rangle}|).
\]
\end{prop}
\begin{proof}
By Proposition~\ref{prop:props} and the definition of $\overline{G}_2(M)$, we see that
\[
G_2(N) = \bigcup_{M\subseteq N} \{ \Phi^M_f \, | \, f \in \overline{L}_2(|M|).
\]
Let $M \subseteq N$ and put $m=|M|$.
Suppose $M \in \Sym$, and $f \in \overline{L}_2(m)$.
If $M^\prime \subseteq N$ and $f^\prime \in \overline L_2(|M^\prime|)$,
then by Proposition~\ref{prop:global-eq}\ref{prop:global-eq-1}
$\Phi^M_f \cong \Phi^{M^\prime}_{f^\prime}$ if and only if
$M \cong M^\prime$ and $f^\prime \in [f]$.
This shows that
\begin{equation}
\label{proof:global-shifts-eq-1}
[\Phi^M_f] = \{ \Phi^{M^\prime}_{f^\prime} \, | \, M^\prime \in [M] \in N / {\cong},
f^\prime \in [f] \in \overline{L}_2(m)  \}
\end{equation}
 is one equivalence class of $G(N) / {\cong}$.
We conclude that there are $|L_2(m) / \mathcal{S}_2| $ equivalence classes whose members are the global maps
that are induced by irreducible local rules and a neighbourhood $M^\prime \in [M]$.

Suppose $M \subseteq N$ and $M \not\in \Sym$.
We first consider the case $f \in \overline{L}_2(m)$, and $\hat{r}f \in \{f, \hat{c}f \}$.
Then $[\Phi^M_f]$ is the same as in Equation~\ref{proof:global-shifts-eq-1}.
The relation $\hat{r}f \in \{f, \hat{c}f \}$ is satisfied, if
$[f] \in t^{\langle r \rangle}_{m}$, $[f] \in t^{\langle rc \rangle}_{m}$,
or $[f] \in t^{\mathcal{S}_2}_{m}$.
Hence we obtain $|t^{\langle r \rangle}_{m}| + |t^{\langle rc \rangle}_{m}|
+ |t^{\mathcal{S}_2}_{m}|$ equivalence classes of global maps.

We now consider the case $f \in \overline{L}_2(m)$, and $\hat{r}f \not\in \{f, \hat{c}f \}$.
Then by Proposition~\ref{prop:global-eq}\ref{prop:global-eq-2}, $[f]$ and $M$ induce the two equivalence classes
\begin{align*}
[\Phi^M_f] & = \{ \Phi^{d+M}_f, \Phi^{d+M}_{\hat{c}f} \, | \, d + M \subseteq N \}
\cup  \{ \Phi^{d-M}_{\hat{r}f}, \Phi^{d-M}_{\widehat{rc}f} \, | \, d  -M \subseteq N \}; \\
[\Phi^M_{\hat{r}f}] & = \{ \Phi^{d+M}_{\hat{r}f}, \Phi^{d+M}_{\widehat{rc}f} \, | \, d + M \subseteq N \}
\cup  \{ \Phi^{d-M}_{f}, \Phi^{d-M}_{\hat{c}f} \, | \, d  -M \subseteq N \}.
\end{align*}
The relation $\hat{r}f \not\in \{f, \hat{c}f \}$ is satisfied, if
$[f] \in t^{\langle 1 \rangle}_{m}$ or $[f] \in t^{\langle c \rangle}_{m}$.
Hence we obtain $2 (|t^{\langle 1 \rangle}_{m}| + |t^{\langle c \rangle}_{m}|)$
equivalence classes of global maps.
We remark that
\[
|t^{\langle r \rangle}_{m}| + |t^{\langle rc \rangle}_{m}| + |t^{\mathcal{S}_2}_{m}|
+ 2 (|t^{\langle 1 \rangle}_{m}| + |t^{\langle c \rangle}_{m}|)
= \overline{L}_2(m) / \mathcal{S}_2 + |t^{\langle 1 \rangle}_{m}| + |t^{\langle c \rangle}_{m}|.
\]
\end{proof}
The equivalence classes of $\mathcal{P}([1,n])/{\cong}$ for small values of $n$ are given by
\begin{align*}
\mathcal{P}([1,0]) & : [\emptyset] \\
\mathcal{P}([1,1]) & : [\emptyset], [\{1\}];  \\
\mathcal{P}([1,2]) & : [\emptyset], [\{1\}],  [\{1,2\}]; \\
\mathcal{P}([1,3]) & : [\emptyset], [\{1\}],  [\{1,2\}], [\{1,3\}], [\{1,2,3\}]; \\
\mathcal{P}([1,4]) & : [\emptyset], [\{1\}],  [\{1,2\}], [\{1,3\}], [\{1,4\}], [\{1,2,3\}], [\{1,2,4\}], [\{1,2,3,4\}]; \\
\mathcal{P}([1,5]) & : [\emptyset], [\{1\}],  [\{1,2\}], [\{1,3\}], [\{1,4\}], [\{1,5\}], [\{1,2,3\}], [\{1,2,4\}], [\{1,2,5\}], \\
        &  [\{1,3,5\}], [\{1,2,3,4\}], [\{1,2,3,5\}], [\{1,2,4,5\}], [\{1,2,3,4,5\}].
\end{align*}
From these equivalence classes, the following are not reflection-symmetrical: $[\{1,2,4\}]$, $[\{1,2,5\}]$, and $[\{1,2,3,5\}]$.
Therefore by Proposition~\ref{prop:global-shift} and Table~\ref{tab:s2_fs}
\begin{align}
|G_2([1,0])/{\cong}| & = |\overline{L}_2(0) / \mathcal{S}_2| = 1; \notag \\
|G_2([1,1])/{\cong}| & = |\overline{L}_2(0) / \mathcal{S}_2| + |\overline{L}_2(1) / \mathcal{S}_2| = 3; \notag \\
|G_2([1,2])/{\cong}| & = |\overline{L}_2(0) / \mathcal{S}_2| + |\overline{L}_2(1) / \mathcal{S}_2|
+ |\overline{L}_2(2) / \mathcal{S}_2| = 7; \notag \\
|G_2([1,3])/{\cong}| & = |\overline{L}_2(0) / \mathcal{S}_2| +  |\overline{L}_2(1) / \mathcal{S}_2|
+ 2 |\overline{L}_2(2) / \mathcal{S}_2|
 + |\overline{L}_2(3) / \mathcal{S}_2| = 85; \label{eq:classes85}\\
|G_2([1,4])/{\cong}| & = |\overline{L}_2(0) / \mathcal{S}_2| + |\overline{L}_2(1) / \mathcal{S}_2|
+ 3 |\overline{L}_2(2) / \mathcal{S}_2| + 2 |\overline{L}_2(3) / \mathcal{S}_2| \notag \\
& + |\overline{t}^{\langle 1 \rangle}_3| + |\overline{t}^{\langle c \rangle}_3|
+ |\overline{L}_2(4) / \mathcal{S}_2| = 16\,656; \notag \\
|G_2([1,5])/{\cong}| & = |\overline{L}_2(0) / \mathcal{S}_2| + |\overline{L}_2(1) / \mathcal{S}_2|
+ 4 |\overline{L}_2(2) / \mathcal{S}_2|+ 4 |\overline{L}_2(3) / \mathcal{S}_2| \notag \\
& + 2|\overline{t}^{\langle 1 \rangle}_3| + 2|\overline{t}^{\langle c \rangle}_3|
+ 3 |\overline{L}_2(4) / \mathcal{S}_2| + |\overline{t}^{\langle 1 \rangle}_4| + |\overline{t}^{\langle c \rangle}_4| \notag \\
& +  |\overline{L}_2(5) / \mathcal{S}_2| = 1\,074\,020\,544 \notag.
\end{align}
The values for n=1,2,3 agree with~\cite{RUIVO2018280}, but the value for n=4 (interval $[1,4]$) differs and merits further investigation.

\subsection{Scaling}

Let us consider a CA $A = (\mathcal{A}, N, f)$, where $N = \{0,2\}$.
If $A$ is initialised with the configuration $x^0$, then after $t$ iterations
of $\Phi^N_f$, the configuration of $A$ is $x^t = (\Phi^N_f)^t(x^0)$.
The value of $x^t_i$ can be calculated recursively as follows:
\[
x^t_i =  f(x^{t-1}_ix^{t-1}_{i+2}) = f(f(x^{t-2}_ix^{t-2}_{i+2}) f(x^{t-2}_{i+2}x^{t-2}_{i+4})  ) = \ldots.
\]
We see that $x^t_i$ depends on $x^{t-1}_i$ and on $x^{t-1}_{i+2}$.
If we go further back until we reach $x^0$, we observe that $x^t_i$ depends on $x^0_i,x^0_{i+2}, \ldots, x^0_{i+2t}$.
If $i$ is even then $x^t_i$ depends only on the even coordinates of $x^0$,
and if $i$ is odd, then it depends only on the odd coordinates.
We construct two configurations $y^0$ and $z^0$, by
selecting either the even or odd coordinates of $x^0$.
We put $y^0 = \ldots x_{-2}.x_0x_2 \ldots$ and $z^0 = \ldots x_{-1}.x_1x_3 \ldots$,
and consider the CA $B = (\mathcal{A}, M, f)$ where $M = \{0,1\}$.
It is easy to see (and will be proven shortly) that if $i=2j$ is even,
then $x^t_{i}= (\Phi^M_f)^t(y^0)_j$, and if
$i=2j+1$ is odd then $x^t_{i} = (\Phi^M_f)^t(z^0)_j$.
We obtain $x^t$
by interleaving (or should we say entangling) $y^t$ and $z^t$,
\[
x^t = \ldots y^t_{-1}z^t_{-1}.y^t_0z^t_0y^t_1z^t_1 \ldots.
\]
The two CAs $A$ and $B$ perform the same computations provided they are initialised accordingly.
Note that we need two copies of $B$ to compute the evolution of $A$, and that, on the other hand,
$A$ can compute two different evolutions of $B$ simultaneously.
Nevertheless, we consider the CAs $A$ and $B$ as equivalent.

Earlier on, we have seen how the isometric geometrical transformations of
translation and reflection on the set of configurations led to symmetry transformations on the rule space.
We now study the uniform scaling of configurations and neighbourhoods.
Let $N_1 = \{j_1, \ldots, j_{n}\} $ be a neighbourhood, let $p$ be a positive integer, and
let $N_2$ be the neighbourhood $N_1$ scaled by the factor $p$, that is
$N_2 = p N_1 = \{p j_1,\ldots,p j_{n}\}$.
We will show how $\Phi_f^{N_2}$ can be reduced to $\Phi_f^{N_1}$.
Define a map $\gamma_p: \mathcal{A}^\mathbb{Z} \to \mathcal{A}^\mathbb{Z}$  by
$\gamma_p(x)_i = x_{pi}$.
The image of the configuration $x$ under $\gamma_p$ contains every $p$-th symbol from $x$.
If $i = qp + r$, $0 \le r < p$, then $\Phi^{N_2}_f(x)_i$ can be expressed as
\begin{multline}
\label{eq:scal}
\Phi^{N_2}_f\bigl(x\bigr)_{qp+r}  =  \Phi^{N_2}_f\bigl(\sigma^r x\bigr)_{pq} =
f\bigl((\sigma^rx)_{pj_1 + pq} \ldots (\sigma^rx)_{pj_{n} + pq}\bigr) \\
 =  f\bigl( \gamma_p(\sigma^rx)_{j_1+q} \ldots \gamma_p(\sigma^rx)_{j_{n}+q}\bigr)
= \Phi^{N_1}_f\bigl(\gamma_p(\sigma^rx)\bigr)_q.
\end{multline}
Based on this relation, we define the
scaling map $s_p : \mathcal{A}^{\mathbb Z} \to (\mathcal{A}^{\mathbb Z})^p$ by
\[
s_p(x) = \bigl(\gamma_p(x), \gamma_p(\sigma x), \ldots , \gamma_p(\sigma^{p-1} x) \bigr).
\]
The map $s_p$ is one-to-one and onto.
Thus the inverse map $s_p^{-1} : (\mathcal{A}^\mathbb{Z})^p \to \mathcal{A}^\mathbb{Z}$ exists.
The preimage of a tuple $(x^{(0)}, \ldots, x^{(p-1)})$ of configurations is
\begin{equation*}
s^{-1}_p(x^{(0)}, \ldots, x^{(p-1)}) =
\ldots x^{(0)}_{-1}x^{(1)}_{-1}\ldots x^{(p-1)}_{-1}.x^{(0)}_0x^{(1)}_0 \ldots x^{(p-1)}_0
x^{(0)}_1x^{(1)}_1 \ldots x^{(p-1)}_1 \ldots.
\end{equation*}
While permutations $\nu \in S_\mathcal{A}$, the reflection $r$, and the shift $\sigma$ of the configuration
are self-homeomorphisms $\mathcal{A}^\mathbb{Z} \to \mathcal{A}^\mathbb{Z}$,
the scaling map $s_p$ is a homeomorphism $\mathcal{A}^\mathbb{Z} \to (\mathcal{A}^\mathbb{Z})^p$.
If $i = qp + r$, $0 \le r < p$, then the $i$th coordinate of the preimage is
\begin{equation}
\label{eq:idx}
s^{-1}_p(x^{(0)}, \ldots, x^{(p-1)})_{qp+r} = x^{(r)}_q.
\end{equation}
Let $\Phi$ be a map on $\mathcal{A}^{\mathbb{Z}}$.
If $y = (x^{(0)}, x^{(1)}, \ldots)$ is a tuple of configurations,
we define
the value of $\Phi(y)$ by elementwise application:
$ \Phi (y) = \Phi(x^{(0)}, x^{(1)}, \ldots) = (\Phi(x^{(0)}), \Phi(x^{(1)}), \ldots)$.
This notation enables us to introduce the scaling operator $\hat{s}_p$
by defining
\[
\hat{s}_p \Phi^{N_1}_f = s^{-1}_p \circ \Phi^{N_1}_f \circ s_p.
\]
Then
\begin{eqnarray*}
\hat{s}_p \Phi^{N_1}_f\bigl(x\bigr) & = & s_p^{-1} \circ \Phi^{N_1}_f \circ s_p \bigl(x\bigr)
= s_p^{-1} \circ \Phi^{N_1}_f \bigl(\gamma_p(x), \ldots, \gamma_p(\sigma^{p-1} x)\bigr) \\
& = & s_p^{-1} \circ  \Bigl(\Phi^{N_1}_f\bigl(\gamma_p(x)\bigr), \ldots, \Phi^{N_1}_f\bigl(\gamma_p(\sigma^{p-1} x)\bigr)
\Bigr). \\
\end{eqnarray*}
From Eq.~(\ref{eq:scal}) and (\ref{eq:idx}) follows
\[
\hat{s}_p \Phi^{N_1}_f\bigl(x\bigr)_{pq+r} = \Phi^{N_1}_f\bigl(\gamma_p(\sigma^{r} x)\bigr)_q
= \Phi^{N_2}_f\bigl(x\bigr)_{qp+r},
\]
which shows that $\Phi^{N_2}_f = \hat{s}_p \Phi^{N_1}_f$.
We try to invert the scaling operator by writing
\[
s_p \circ \Phi^{N_2}_f \circ s^{-1}_p.
\]
The resulting map, however, is not a global map on $\mathcal{A}^\mathbb{Z}$,
but a map on $(\mathcal{A}^\mathbb{Z})^p$.
A tuple of configurations $(x^{(0)}, \ldots, x^{(p-1)})$ is processed at once:
\[
\bigl(s_p \circ \Phi^{N_2}_f \circ s^{-1}_p \bigr) \bigl(x^{(0)}, \ldots, x^{(p-1)}\bigr)
= \bigl(\Phi^{N_1}_f(x^{(0)}), \ldots, \Phi^{N_1}_f(x^{(p-1)})\bigr).
\]
The $p$ configurations are merged together into one configuration by the map $s^{-1}_p$, the
CA $(\mathcal{A}, N_2, f)$ processes the one configuration,
and the successor configuration is then split up again by $s_p$ into $p$ configurations.
This is consistent with the example above where the CA $A$ simultaneously processes
two configurations of CA $B$.
To obtain an operational definition of the inverse scaling operator, we
therefore choose the following approach.
If there is a positive integer $p$ such that $N = pM$, we define
\[
\hat{s}^{-1}_p \Phi^{N}_f (x) = \gamma_p \circ \Phi^{N}_f \circ s^{-1}_p(x, \ldots, x).
\]
The initial configuration $x$ is copied $p$ times to obtain the tuple $(x, \ldots, x)$.
The tuple is merged together by the map $s^{-1}_p$, processed by the CA $(\mathcal{A}, N, f)$, and
then only every $p$th coordinate is selected, which results in the target configuration.
By the definition of the operators, we see that
\[
\Phi^{N_2}_f = \hat{s}_p \Phi^{N_1}_f \ \ \ \mbox{ if and only if} \ \ \
\Phi^{N_1}_f = \hat{s}^{-1}_p \Phi^{N_2}_f.
\]

We extend the scaling operator to rational numbers $v=p/q$ provided that $v N_1$ is a set of integers, by
defining
\[
\hat{s}_v \Phi^{N_1}_f = \hat{s}^{-1}_q \hat{s}_p \Phi^{N_1}_f.
\]
If $p$ is an integer, the operator $\hat{s}_p$ is a map on $\mathcal{C}_A$ and is defined for all CAs.
It is one-to-one, but not onto, with the exception of $\hat{s}_1$, which is the identity map.
If $v$ is a positive rational, but not an integer, the operator $\hat{s}_v$  is only a partial map on $\mathcal{C}_A$.

Earlier on we defined the congruence relation $\cong$ on $\mathbb{Z}$.
We now define a similarity relation $\sim$ on $\mathbb{Z}$ that takes into account the scaling of neighbourhoods.
Let $M \subseteq \mathbb{Z}$ and $N \subseteq \mathbb{Z}$.
Then $M \sim N$ if and only if there is a set $M^\prime \subseteq \mathbb{Z}$,
and a set $N^\prime \subseteq \mathbb{Z}$, and a positive rational $v$ such that
$M \cong M^\prime$, $N \cong N^\prime$, and $N^\prime = v M^\prime$.
Note that $M \cong N$ implies $M \sim N$ (choose $v=1$).

In a similar manner, we define the equivalence relation $\sim$ on $\mathcal{C}_\mathcal{A}$.
If $\Phi^A$ and $\Phi^B$ are two global maps on $\mathcal{A}^\mathbb{Z}$,
then $\Phi^A \sim \Phi^B$ if and only if
there are global maps $\Phi^{A^\prime}$ and $\Phi^{B^\prime}$, and a positive rational $v$,
such that $\Phi^A \cong \Phi^{A^\prime}$, $\Phi^B \cong \Phi^{B^\prime}$,
and $\Phi^{B^\prime} = \hat{s}_v \Phi^{A^\prime}$.

We give an example.
Consider the local rule $W^3_{90}$ given by $W^3_{90}(a_0a_1a_2) = a_0 \oplus a_2$,
where the symbol $\oplus$ denotes the exclusive-or logical operation.
Since $W^2_6(b_0b_1) = b_0 \oplus b_1$, we have
$\Phi^{[-1,1]}_{W^3_{90}} = \Phi^{\{-1,1\}}_{W^2_6}$.
The neighbourhood $\{-1,1\}$ is not the multiple of another neighbourhood, but
$\{0,2\}$ is, since $\{0,2\} = 2 \{0,1\}$.
Therefore
\[
\Phi^{[-1,1]}_{W^3_{90}} = \sigma^{-1} \Phi^{[0,2]}_{W^3_{90}} =
\sigma^{-1} \Phi^{\{0,2\}}_{W^2_6} = \sigma^{-1} \hat{s}_2 \Phi^{[0,1]}_{W^2_6},
\]
and
\[
\Phi^{[-1,1]}_{W^3_{90}} \sim \Phi^{[0,1]}_{W^2_6} = \Phi^{[-1,1]}_{W^3_{60}}.
\]

The next statement follows from a slight adaptation of Proposition~\ref{prop:global-shift} and its proof.
\begin{prop}
\label{prop:global-scaling}
Suppose $N = [1, n]$ is an integer interval. Then
\[
|G_2(N)/{\sim}| = \sum_{[M]\in \mathcal{P}(N)/{\sim}} |\overline{L}_2(|M|)/\mathcal{S}_2| +
      \sum_{[M]\in \mathcal{P}(N)/{\sim}, M \not\in Sym} (|t_{|M|}^{\langle 1 \rangle}| + (|t_{|M|}^{\langle c \rangle}|).
\]
\end{prop}
Compared to the relation $\cong$, the number of equivalence classes remains unchanged
if $N = [1,0]$, $[1,1]$, or $[1,2]$.
The next two results take into account that $\{1,2\} \sim \{1,3\} \sim \{1,4\} $:
\begin{align}
|G_2([1,3])/{\sim}| & = \sum_{m=0}^3 |\overline{L}_2(m) / \mathcal{S}_2| =
|G_2([1,3])/{\cong}| - |\overline{L}_2(2) / \mathcal{S}_2| = 81; \label{eq:classes81}  \\
|G_2([1,4])/{\sim}| & = \sum_{m=0}^2 |\overline{L}_2(m) / \mathcal{S}_2|
+ 2 |\overline{L}_2(3) / \mathcal{S}_2|
+ |\overline{t}^{\langle 1 \rangle}_3| + |\overline{t}^{\langle c \rangle}_3| +  |\overline{L}_2(4) / \mathcal{S}_2| \notag \\
& = |G_2([1,4])/{\cong}| - 2 |\overline{L}_2(2) / \mathcal{S}_2| = 16\,648 \notag.
\end{align}
\begin{prop}
\label{commutators}
Let $\hat{\nu}$ be any of the permutation operators and let $v=p/q$ be a positive rational. Then
\begin{multicols}{3}
    \begin{enumerate}[label=(\roman*)]
        \item \label{comm:pi-s} $\hat{\nu} \hat{s}_v = \hat{s}_v \hat{\nu}$
        \item \label{comm:r-s} $\hat{r} \hat{s}_v = \hat{s}_v \hat{r}$
        \item \label{comm:sigma-s} $\sigma^p \hat{s}_v = \hat{s}_v \sigma^q$.
    \end{enumerate}
    \end{multicols}
\end{prop}

\begin{proof}

\ref{comm:pi-s}
The permutation of the state set is independent of the geometrical transformation of the neighbourhood, which
implies the asserted statements.

\ref{comm:r-s}
Suppose that $N$ is a neighbourhood and $vN$ is also a set of integers.
Then
\[
\hat{r} \hat{s}_v \Phi^N_f = \hat{r} \Phi^{vN}_f = \Phi^{-vN}_{\hat{r}f}
= \hat{s}_v \Phi^{-N}_{\hat{r} f} = \hat{s}_v \hat{r} \Phi^N_f.
\]

\ref{comm:sigma-s}
We first assume that $v=p$ is a positive integer.
Then
\[
\sigma^p \hat{s}_p \Phi^N_f = \sigma^p \Phi^{pN}_f = \Phi^{p(N+1)}_f = \hat{s}_p \Phi^{N+1}_f = \hat{s}_p \sigma \Phi^N_f.
\]
As consequence, the relation for the inverse operator reads $\hat{s}^{-1}_q \sigma^q = \sigma \hat{s}^{-1}_q$.
Hence
\[
\sigma^p \hat{s}_{p/q} = \sigma^p \hat{s}_p \hat{s}^{-1}_q = \hat{s}_p \sigma \hat{s}^{-1}_q
= \hat{s}_p \hat{s}^{-1}_q \sigma^q = \hat{s}_{p/q} \sigma^q.
\]
\end{proof}
The next proposition summarizes some of the statements we have already encountered.
\begin{prop}
If $N = [1,n]$, then
\begin{enumerate}[label=(\roman*)]
\item
$
|\overline{L}_\mathcal{A}(n)/\mathcal{S}_\mathcal{A}| = |\overline{G}_\mathcal{A}(N) / \hat{\mathcal{S}}_\mathcal{A}(n+1)|
= |\overline{G}_\mathcal{A}(N) / {\sim}| = |\overline{G}_\mathcal{A}(N) / {\cong}|;
$
\item
$
|L_\mathcal{A}(n)/\mathcal{S}_\mathcal{A}| = |G_\mathcal{A}(N) / \hat{\mathcal{S}}_\mathcal{A}(n+1)|;
$
\item
$
|\overline{G}_\mathcal{A}(N) / \hat{\mathcal{S}}_\mathcal{A}(n+1)| \le
|G_\mathcal{A}(N) / {\sim}| \le
|G_\mathcal{A}(N) / {\cong}| \le
|G_\mathcal{A}(N) / \hat{\mathcal{S}}_\mathcal{A}(n+1)|.
$
\end{enumerate}
\end{prop}

\subsection{Elementary Cellular Automata}
Recall that elementary CAs are precisely those binary CAs with the neighbourhood $N = [-1,1]$,
and that their global rule space is denoted by $G_2(N)$.
We set $\mathcal{E} = G_2(N)$.
In this subsection we examine in more detail the partitions of $\mathcal{E}$ that arise from various equivalence relations.
It is well-known that under reflection and permutation,  the set $\mathcal{E}$ decompose into 88 equivalence classes,
in symbols, $|\mathcal{E} / \mathcal{S}_2| = 88$.
All equivalence relations discussed here will be coarser than the partition $\mathcal{E} / \mathcal{S}_2$.
In general, if $R$ and $S$ are equivalence relations on a set $A$, and $R$ is coarser than $S$,
then each equivalence class $C$ of $A / R$ is either itself an equivalence class of $A / S$ or the union of
two or more equivalence classes from $A / S$.
We specify the members of equivalence classes using their Wolfram codes.
Whenever a number $c$ appears in one of the equivalence classes below, it is to be interpreted as the global map $\Phi^{N}_{W^3_c}$.
Every equivalence relation considered in what follows is obtained by merging two or more of the following equivalence classes of
$\mathcal{E} / \mathcal{S}_2$ into larger ones.
\begin{align*}
& A_1 = \{3,17,63,119\}, & & A_2=\{5,95\}, & & A_3 = \{10,80,175,245\}, \\
& A_4 = \{12,68,207,221\}, & & A_5=\{15,85\}, & & A_6=\{34,48,187,243\}, \\
& A_7 = \{51\}, & & A_8=\{60,102,153,195\}, & & A_9 = \{90, 165\}, \\
& A_{10} = \{136, 192, 238, 252\}, & & A_{11}=\{160,250\}, & & A_{12}=\{170,240\}, \\
& A_{13}=\{204\}, & &  A_{14} = \{77\},  & & A_{15} = \{232 \}, \\
& A_{16} = \{ 23 \}, & & A_{17} = \{178\}, & & A_{18}=\{105\}, \\
& A_{19} = \{150\}. & & & &
\end{align*}
The first further equivalence relation that we will consider is the congruence relation $\cong$ that
adds the shift operator to permutation and reflection.
By Eq.~(\ref{eq:classes85}) this equivalence relation induces 85 equivalence classes.
The equivalence classes of $\mathcal{E} / \mathcal{S}_2$ that are merged to obtain the coarser equivalence classes of
$\mathcal{E} / {\cong}$ are the following, see also~\cite{RUIVO2018280}.
\begin{align*}
& B_1 = A_4 \cup A_6, & & B_2 = A_5 \cup A_7, & & B_3 = A_{12} \cup A_{13}.
\end{align*}
Next, we consider the equivalence relation $\sim$ on $\mathcal{E}$, which adds scaling to the set of
symmetry operations.
By Eq.~(\ref{eq:classes81}) the number of equivalence classes decreases to 81.
The equivalence classes of $\mathcal{E} / {\sim}$ that do not coincide with the equivalence classes of
$\mathcal{E} / \mathcal{S}_2$ are the following.
\begin{align*}
& C_1 = A_3 \cup A_4 \cup A_6, & & C_2=A_5 \cup A_7, & & C_3=A_8 \cup A_9,  \\
& C_4 = A_{10} \cup A_{11}, & & C_5 = A_1 \cup A_2, & & C_6 = A_{12} \cup A_{13}.
\end{align*}
Epperlein \cite{epperlein2015classification, epperlein2017}
classified the elementary CAs up to topological conjugacy.
Permutation and reflection are general topological conjugacies and yield
the well-known 88 equivalence classes.
Epperlein demonstrated that, in addition, there exist further topological conjugacies between
specific pairs of elementary CAs, and that apart from these there are no others.
The pairs
\begin{align*}
& ([15], [170]), & & ([23], [178]), & & ([77],[232])
\end{align*}
of equivalence classes of $\mathcal{E} / \mathcal{S}_2$ are conjugated
by
\[
\vartheta: \{0,1\}^\mathbb{Z} \to \{0,1\}^\mathbb{Z},  \ \
\vartheta(x)_i =
\begin{cases}
cx_i & \text{if $i$ is odd,} \\
x_i & \text{if $i$ is even.}
\end{cases}
\]
The equivalence classes $[90]$, $[105]$, and $[150]$ of $\mathcal{E} / \mathcal{S}_2$ are
conjugated to each other due to a result by Kurka~\cite{Kurka2003TopologicalAS}.
Let us denote the equivalence classes of the elemenary CAs that are obtained by taking
all topological conjugacies into account by $\mathcal{E} / \mathrm{TC}$.
Consequently, the equivalence classes of $\mathcal{E} / \mathrm{TC}$ that differ from equivalence classes of $\mathcal{E}/ \mathcal{S}_2$
are the following.
\begin{align*}
& D_1 = A_5 \cup A_{12}, & & D_2 = A_{14} \cup A_{15}, & & D_3 = A_{16} \cup A_{17},
& & D_4 = A_{9} \cup A_{18} \cup A_{19}.
\end{align*}
The partition $\mathcal{E}/\mathrm{TC}$ thus contains 83 equivalence classes.
Finally, we consider the equivalence classes obtained from $\mathcal{E} / \mathrm{TC}$
when we additionally incorporate shifting and scaling (note that permutation and reflection are already topological conjugacies).
We obtain these by taking the transitive closure of $\mathcal{E} / {\sim}$ and $\mathcal{E} / \mathrm{TC}$.
Let $\mathcal{E} / \Omega$ denote the resulting partition.
Then the equivalence classes of $\mathcal{E}/ \Omega$ that are distinct from those of
$\mathcal{E} / \mathcal{S}_2$ are the following.
\begin{align*}
& E_1 = A_3 \cup A_4 \cup A_6, & & E_2=A_5 \cup A_7 \cup A_{12} \cup A_{13}, & & E_3=A_8 \cup A_9 \cup A_{18} \cup A_{19}, \\
& E_4 = A_{10} \cup A_{11}, & & E_5 = A_1 \cup A_2, & & E_6 = A_{14} \cup A_{15}, \\
& E_7 = A_{16} \cup A_{17}. & & & &
\end{align*}
We see that 19 equivalence classes of $\mathcal{E} / \mathcal{S}_2$ are merged into
7 equivalence classes of $\mathcal{E} / \Omega$.
The other equivalence classes of $\mathcal{E} / \Omega$ are also equivalence classes of $\mathcal{E} / \mathcal{S}_2$.
Consequently, $\mathcal{E} / \Omega$ contains $88 - 19 + 7 = 76$ equivalence classes.

\section{Extended Symmetries and the Curtis--Hedlund--Lyndon Framework}
\label{sec:extended}

The classification developed in Sections~\ref{sec:loc} and \ref{sec:glob} relies on a specific group of symmetry operators, $\mathcal{T}_{\mathcal{A}}$, generated by state permutations, lattice shifts, and reflections. Combined with the scaling operator $\hat{s}_v$, these transformations constitute the \emph{geometric} symmetries of cellular automata. In this section, we briefly discuss how this classification relates to broader notions of equivalence in topological dynamics, specifically topological conjugacy and the Curtis--Hedlund--Lyndon (CHL) framework.

\subsection{Implementing vs. Preserving Symmetries}
We recall the Curtis--Hedlund--Lyndon characterization: a map $\Phi : A^{\mathbb{Z}} \to A^{\mathbb{Z}}$ is the global map of a cellular automaton if and only if it is continuous and commutes with the shift $\sigma$. Such maps (sliding block codes) are the morphisms of the symbolic dynamical system. A symmetry $S$ can relate to this framework in three ways:
\begin{itemize}
  \item \textbf{CHL-implementing:} $S$ itself acts on configurations $A^{\mathbb{Z}}$ as a sliding block code. Standard examples are the shift operator $\sigma$ and state permutations $\hat{\nu}$.
  \item \textbf{CHL-preserving:} $S$ is a transformation of the space of global maps that sends CAs to CAs, but $S$ itself may not be a valid CA global map. The primary example is the reflection operator $\hat{r}$, defined by $(\hat{r}x)_i = x_{-i}$. While $\hat{r}$ is a homeomorphism, it anti-commutes with the shift ($\hat{r}\sigma = \sigma^{-1}\hat{r}$) and thus is not a CA. However, conjugating a CA by $\hat{r}$ yields another CA.
  \item \textbf{Rule-space operator:} $S$ acts directly on local rules but has no simple expression as a configuration map. Diagonal shifts (e.g., adding a constant to the rule output) fall into this category.
\end{itemize}

\subsection{Topological Conjugacy}
In symbolic dynamics, two global maps $\Phi$ and $\Psi$ are \emph{topologically conjugate} if there exists a homeomorphism $h: \mathcal{A}^{\mathbb{Z}} \to \mathcal{A}^{\mathbb{Z}}$ such that $h \circ \Phi = \Psi \circ h$. If $h$ also commutes with the shift, the CHL theorem implies that $h$ is a reversible cellular automaton.
Two of the symmetry operators, permutation and reflection, are topological conjugacies.
The shift is on the one hand a conjugacy as $\sigma \circ \Phi = \Phi \circ \sigma$,
on the other hand, it can also be seen as an operator $\mathcal{C}_\mathcal{A} \to \mathcal{C}_\mathcal{A}$.
In this second meaning, it is no longer a conjugacy.
The scaling map $s_p$ is a conjugacy between the dynamical systems
 $((\mathcal{A}^\mathbb{Z})^p, \Phi)$ and $(\mathcal{A}^\mathbb{Z}, \hat{s}_p \Phi)$.
Consequently,
the equivalence relation induced by all topological conjugacies $A^\mathbb{Z} \to A^\mathbb{Z}$
and the equivalence relation induced by all symmetry operators are not comparable, that is,
neither is finer or coarser than the other.

\subsection{Dynamical Extensions}
The group $\mathcal{T}_{\mathcal{A}}$ acts on the space of global maps $\mathcal{C}_{\mathcal{A}}$ primarily by transforming the spatial domain. It does not account for symmetries involving time or information loss:
\begin{enumerate}
    \item \textbf{Time Reversal:} If a CA $\Phi$ is reversible, its inverse $\Phi^{-1}$ is also a CA. The equivalence relation defined by $\mathcal{T}_{\mathcal{A}}$ does not generally equate $\Phi$ with $\Phi^{-1}$.
    \item \textbf{Coarse-graining:} Unlike the exact scaling operator $\hat{s}_p$, coarse-graining transformations (often used in renormalization group analysis) involve projections $\pi: A^p \to B$ that may lose information. These establish semi-conjugacies rather than the strict equivalences considered here.
\end{enumerate}

\subsection{Rigidity of Geometric Symmetries}
One might ask why we restrict attention to shifts, reflections, and permutations. A rigidity phenomenon appears if we focus on geometric index-symbol relabelings of the form $(Hx)_i = \pi(x_{\tau(i)})$ and demand that conjugation by $H$ always sends CAs to CAs. Under the CHL constraints (locality, translation invariance), it can be shown that the index map $\tau$ must be affine ($\tau(i) = ai+b$).
Thus, the only geometric symmetries that strictly preserve the CA class
if the neighbourhood has at least two elements, are built from shifts, reflections, state permutations, and (with suitable encodings) scalings.
If the neighbourhood is a singleton, then $\tau$ can be any bijection, in
accordance with the fact that an affine transformation is determined by at least two points.

Conceptually, extending the notion of symmetry to allow \emph{any} computable map on rule space would trivialise the classification, collapsing distinct dynamical behaviours into single classes. The CHL-compatible symmetries form a principled, resource-constrained class of transformations that keep the symmetry concept meaningful and nontrivial.

\section{Summary}
\label{sec:sum}

In this work, we have established a comprehensive framework for classifying one-dimensional cellular automata through the lens of symmetry. By rigorously defining the actions of state permutation, lattice reflection, translation (shift), and neighbourhood scaling, we have derived a hierarchy of equivalence relations that progressively reduce the complexity of the cellular automaton rule space.

A central insight of this study lies in the operational nature of the four symmetries considered: permutation, reflection, scaling, and shift.
The operations of \textit{state permutation} and \textit{lattice reflection} act as topological conjugacies. If two global maps $\Phi$ and $\Psi$ are related by such a symmetry $h$, they satisfy the conjugacy relation $h \circ \Phi = \Psi \circ h$. Consequently, the dynamics of $\Phi$ can be fully reconstructed from $\Psi$ via the relation $\Phi = h^{-1} \circ \Psi \circ h$. This extends to arbitrary time steps $t$, where the $t$-th iterate is given by
\[
\Phi^t = h^{-1} \circ \Psi^t \circ h.
\]
This property is fundamental: it allows one to transform an initial configuration, evolve it using the equivalent CA $\Psi$ for $t$ steps, and transform the result back to obtain the state of $\Phi$ at time $t$.

\textit{Scaling}, however, differs structurally from these isometries. It does not constitute a topological conjugacy from $A^{\mathbb{Z}}$ to $A^{\mathbb{Z}}$ in the standard sense. Instead, it relates the dynamics on the original lattice to a decimated dynamics on a coarse-grained lattice, mapping $A^{\mathbb{Z}} \to (A^{\mathbb{Z}})^p$. While not a standard automorphism of the shift dynamical system, it represents a meaningful computational equivalence based on interleaving and decimation.

The \textit{shift} operator occupies a unique position in this framework. While it is not a topological conjugacy between distinct rules in the static sense (a rule commutes with the shift rather than being transformed by it), it is strictly necessary to define the symmetry group $\mathcal{T}_{\mathcal{A}}$ fully. Specifically, the shift allows for a unified definition of reflection symmetries regardless of whether the center of symmetry falls on a cell or between cells. Furthermore, this work emphasizes the necessity of defining these operations on \textit{global maps} rather than local rules.
Previous approaches in connection with equivalence classes (e.g., Ruivo et al.~\cite{RUIVO2018280}) defined shifts on local rules, which is theoretically restrictive; one cannot shift a finite neighbourhood indefinitely without encountering boundary issues or redefining the rule's domain. By lifting the analysis to global maps $\Phi: A^{\mathbb{Z}} \to A^{\mathbb{Z}}$, the shift becomes a well-defined automorphism $\sigma$, and the equivalence relation becomes robust.

Methodologically, the concept of \textit{irreducibility} proved to be the canonical tool for navigating these symmetries. We demonstrated that every global map admits a unique representation by an irreducible local rule. This bijection allowed us to apply combinatorial counting techniques---specifically the Principle of Inclusion-Exclusion and classification by stabilizer types---to derive exact formulae for the number of distinct cellular automata.

Our quantitative results for binary cellular automata include:
\begin{itemize}
    \item A derivation of the number of irreducible local rules for any arity $n$.
    \item A complete classification of irreducible binary rules by their stabilizer subgroups (types) under the Klein four-group action of permutation and reflection.
    \item The enumeration of equivalence classes for global maps induced by contiguous neighbourhoods.
Notably, for the elementary cellular automata (neighbourhood size 3), the inclusion of the scaling operator reduces the classical count of 88 equivalence classes (under permutation and reflection) to \textit{81 equivalence classes}.
The number decreases to \textit{76 equivalence classes} if the remaining topological conjugacies
between specific pairs of elementary cellular automata are taken into account.
\end{itemize}

Ultimately, this symmetry-based classification serves as a coarse-graining of the rule space that complements
distinctions such as topological conjugacy.
By identifying rules that are structurally identical up to coordinate transformation, decimation, or relabeling,
we provide a clearer map of the universe of one-dimensional cellular automata.

\textbf{Credit authorship contribution statement:}
Martin Schaller: Writing---original draft, Visualization, Software, Investigation, Formal analysis, Conceptualization. Karl Svozil: Methodology, Investigation, Conceptualization.

\textbf{Declaration of competing interest:}
The authors declare that they have no known competing financial interests or personal relationships that could have appeared to influence the work reported in this paper.

\textbf{Data Availability:} No data was used for the research described in the article.

\textbf{Acknowledgements:}
This research was funded in whole or in part by the Austrian Science Fund (FWF) [Grant DOI:10.55776/PIN5424624]. The authors acknowledge TU Wien Bibliothek for financial support through its Open Access Funding Programme.

\bibliographystyle{plain}
\bibliography{mybib}

\end{document}